\documentclass[USenglish]{article}

\usepackage{newstyle}

\begin{document}
% Author macros::begin %%%%%%%%%%%%%%%%%%%%%%%%%%%%%%%%%%%%%%%%%%%%%%%%
\title{Combinatorial Properties of Self-Overlapping Curves and Interior Boundaries\thanks{Both authors were supported by NSF grant CCF-1618469. Parker Evans was supported by a Goldwater Scholarship from the Goldwater Foundation. An extended abstract of this paper has been accepted for publication at SoCG \cite{efw20}.}}

\author{Parker Evans\footnote{Department of Mathematics, Rice University, Houston, TX, USA; Parker.G.Evans@rice.edu}
  \and
Carola Wenk\footnote{Department of Computer Science, Tulane University, New Orleans, LA, USA;cwenk@tulane.edu}}

\date{}

\maketitle

\begin{abstract}
We study the interplay between the recently defined
concept of \emph{minimum homotopy area}
%\cite{CW13}
and the classical topic of \emph{self-overlapping curves}.  The latter are
plane curves which are the image of the boundary of an immersed disk.
Our first contribution is to prove new
sufficient combinatorial conditions for a curve to be \SO.
%
%We show, loosely speaking, that a curve $\gamma$ without any
%\SO\ subcurves is \SO.
We show that a curve~$\gamma$ with Whitney index 1 and without any
\SO\ subcurves is \SO.
As a corollary, we obtain sufficient conditions for
\SO ness solely in terms of the Whitney index
%, or turning index,
of the curve and its subcurves.
These results follow from our second contribution, which shows that
any plane curve $\gamma$, modulo a basepoint condition, is transformed
into an \emph{interior boundary} by wrapping around $\gamma$ with
Jordan curves. Equivalently, the minimum homotopy area of $\gamma$ is
reduced to the minimal possible threshold, namely the winding area, through 
wrapping. In fact, we show that $n+1$ wraps suffice, where $\gamma$ has $n$ vertices.
Our third contribution is to prove the equivalence
of various definitions of \SO\ curves and interior boundaries, often
implicit in the literature. We also introduce and characterize \emph{zero-obstinance
curves}, further generalizations of interior boundaries defined
by optimality in minimum homotopy area.
\end{abstract}

\section{Introduction}\label{sec:intro}
%\subsection{Sufficient Combinatorial Conditions for Self-Overlapping Curves}
Classically, a curve $\gamma: \mathbb{S}^1 \rightarrow \mathbb{R}^2$ is called
\myemph{\SO} if there is an orientation-preserving immersion $F:
\mathbb{D}^2 \rightarrow \mathbb{R}^2$ of the unit disk $\mathbb{D}^2$, a map of full rank on the entire unit
disk $\mathbb{D}^2$, such that $F|_{\partial \mathbb{D}^2} = \gamma$.
One can think of such an immersion as distorting a unit
disk that lies flat in the plane and stretching and pulling it continuously   without leaving the plane
and without twisting or pinching it~\cite{Immersion}.
If the disk is painted blue on top and pink on the
bottom, then one only sees blue. 
\begin{figure}[htbp]
\centering
    \includegraphics[width=.35\textwidth]{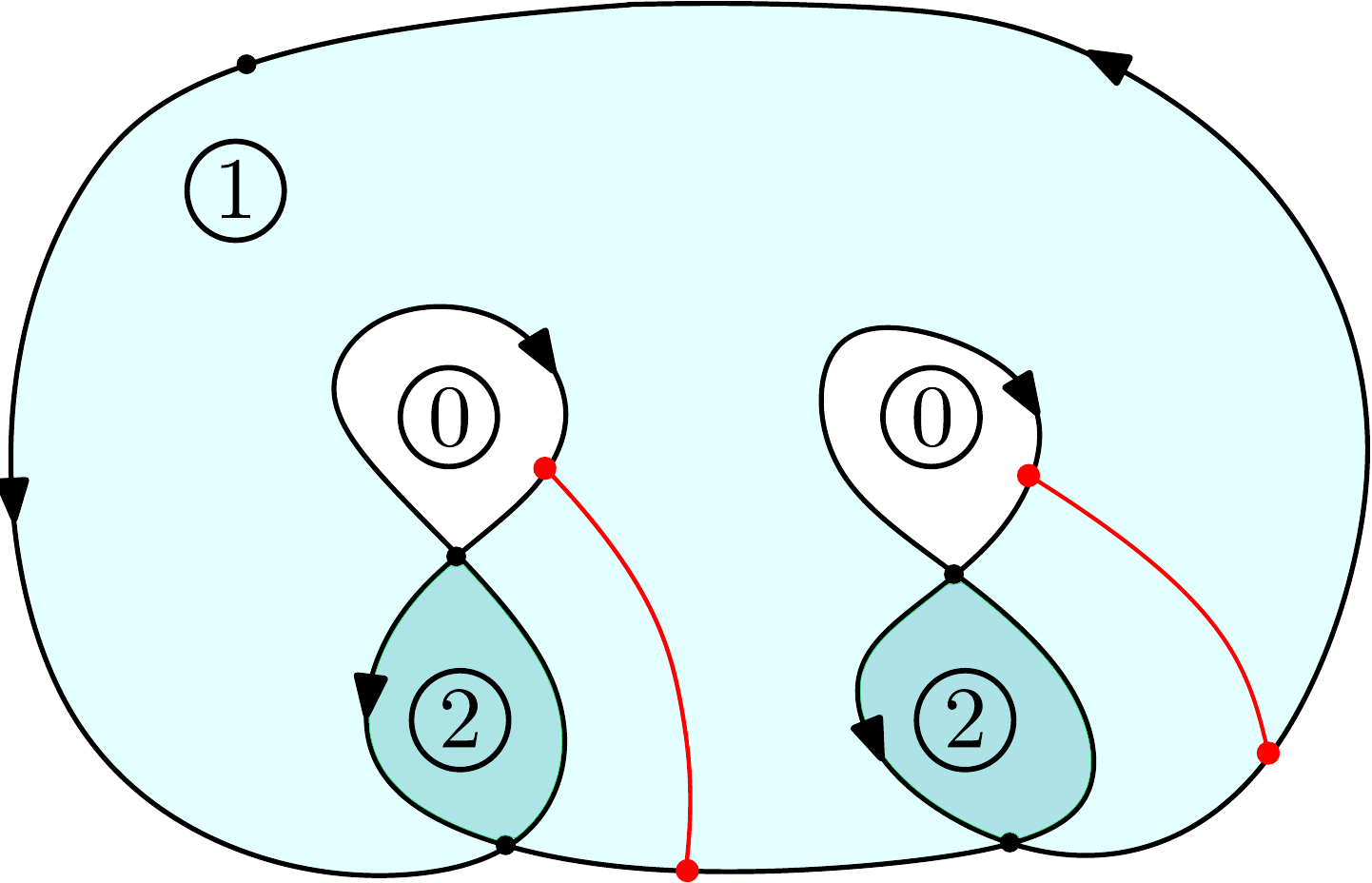}
    \caption{A \SO\ curve $\gamma$ with
      winding numbers for the faces circled. The Blank cuts, shown in red, slice $\gamma$
      into a collection of simple positively oriented (counterclockwise) Jordan~curves.
      }
    \label{fig:NiceSO}
\end{figure}
If we also imagine the
disk being semi-transparent, then
the blue will appear darker
in the regions where it overlaps itself; see \figref{NiceSO}.
That means, any \SO \,
curve~$\gamma$ must have non-negative winding numbers, $wn(x, \gamma) \geq
0$ for every~$x \in \R^2$. We call this condition \myemph{positive
  consistent}.
Another simple
and intuitive view originates from Blank \cite{BL67}: The curve is
\SO\ when we can cut it along simple curves into simple
positively oriented Jordan curves, i.e., a collection
of blue topological disks.
\myemph{Interior boundaries} are generalizations of \SO\ curves that
are defined similarly, except that $F$ is an interior map which allows
finitely many branch points \cite{MarxBranchpoints}.
%are defined similarly, except that finitely many branch points of the
%map $F$ are allowed.
%
Interior boundaries are  composed of multiple \SO\
curves (of the same orientation); see \figref{smallerCurveLattice} for an example.
In this paper, all curves $\gamma: [0,1] \rightarrow \R^2$ are
assumed to be closed, immersed, and generic, i.e., with only finitely
many intersection points, each of which are transverse double
points. We also assume $\gamma'(t)$ exists and is nonzero for all
$t\in[0,1]$.
%
%% Self-overlapping curves have been studied under the lenses of
%% analysis, topology, geometry, combinatorics, as well as graph theory,
%% often requiring a blend of techniques.  Consequently, various
%% perspectives exist on how one conceptualizes \SO\ curves.
%
We show new combinatorial properties of \SO\ curves and interior boundaries
by revealing new connections to the minimum homotopy area of curves.

\begin{figure}[hbt]
\centering
    \includegraphics[width=\textwidth]{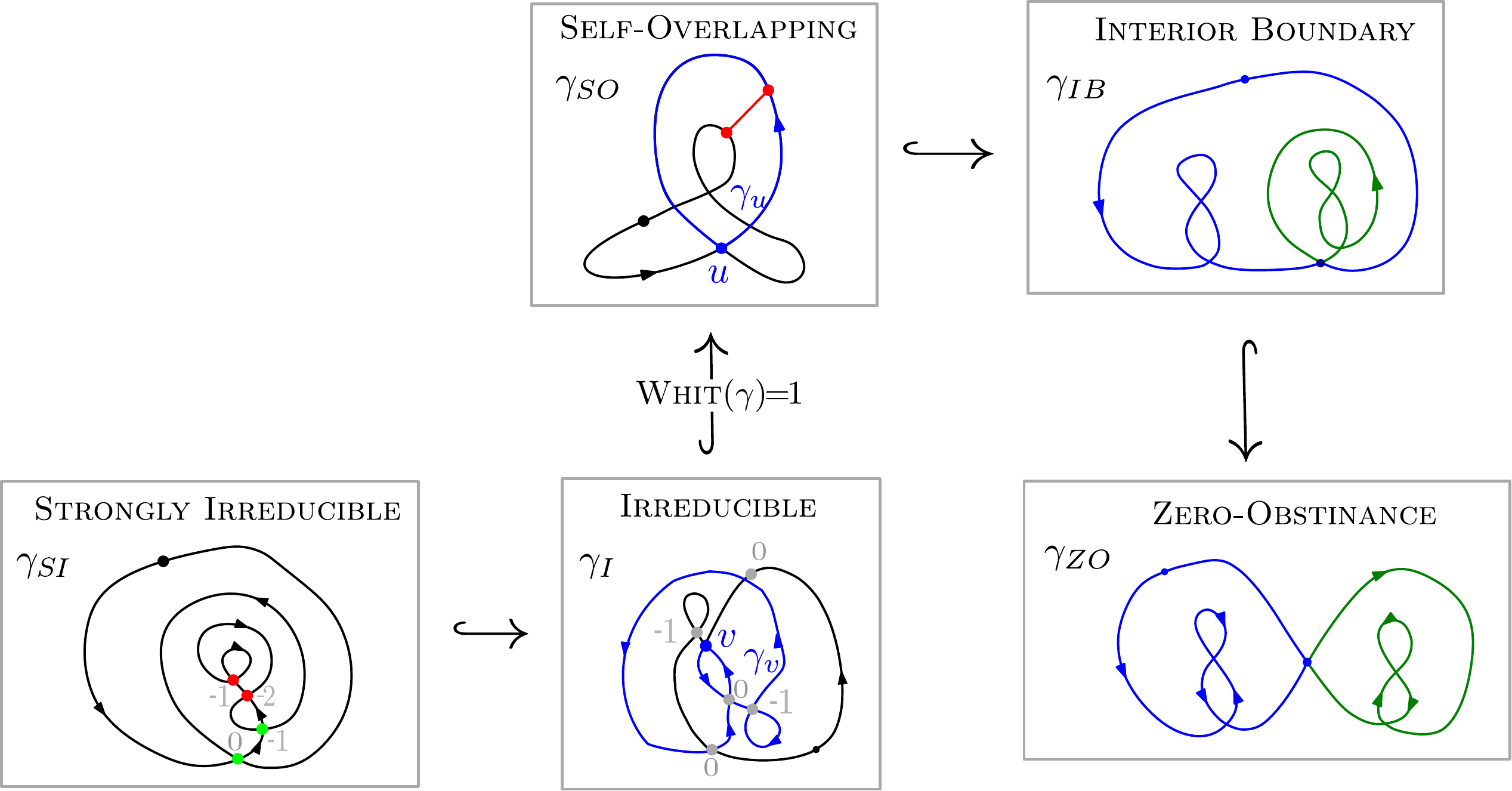}
    \caption{Example curves of different curve classes and inclusion relationships between the
      classes.
       $\gamma_{SO}$ is \SO\ as indicated by the Blank cuts in red.
      $\gamma_{IB}$ is an interior boundary consisting of two
      \SO\ curves (of the same orientation), one in blue the other in
      green.
      The bottom row shows curve classes that are introduced in this paper:
      $\gamma_{SI}$ is strongly irreducible as can be seen from the
      non-positive Whitney indices (shown in gray) of its direct
      split subcurves. Similarly, $\gamma_I$ is irreducible; note that
      $\gamma_{v}$ has Whitney index $1$ but is not \SO.
      Also note that $\gamma_{SO}$ is not irreducible since $\gamma_u$ is \SO.
      $\gamma_{ZO}$ also consists of two \SO\ curves but of different orientation and is therefore not an interior boundary, but it has zero obstinance.
    }
   \label{fig:smallerCurveLattice}
\end{figure}

\subsection{Related work}
\paragraph{Self-Overlapping Curves and Interior Boundaries.}
Self-overlapping curves and interior boundaries have
a rather rich history, and have been studied
%in mathematical literature,
under the lenses of analysis, topology, geometry, combinatorics, and graph
theory~\cite{BL67,Eppstein,Graver,MarxBranchpoints,Marx,Mukherjee,Immersion,SVW92,TT61}.
In the 1960s, Titus \cite{TT61} provided the first algorithm to test whether a curve is \SO\ (or an interior boundary),
by defining a set of cuts that must cut the curve into smaller subcurves that are \SO\ (or interior boundaries).
%He also introduced intersection sequences as an algebraic way to formulate his algorithm.
%
In a 1967 PhD thesis \cite{BL67}, Blank proved that a curve is \SO\ iff there
is a sequence of cuts (different from Titus cuts) that completely
decompose the curve into simple pieces. He represents plane curves with words
and showed that one can determine the existence
of a cut decomposition by looking for algebraic decompositions of the
word. In the 1970s, Marx \cite{Marx} extended Blank's work to give an algorithm to test if a curve is an interior boundary.
In the 1990s, Shor and Van Wyk \cite{SVW92} expedited Blank's algorithm
to run in~$O(N^3)$ time for a polygonal curve with $N$ line segments. Their dynamic programming algorithm is currently the fastest algorithm to test for \SO ness. It is not known whether this runtime bound is tight or whether a faster runtime might be achievable.
In distantly related work, Eppstein and Mumford \cite{Eppstein} showed
that it is NP-complete to determine whether a fixed \SO\ curve $\gamma$
is the 2D projection of an immersed surface in $\R^3$ defined over a compact two-manifold with boundary.
Graver and Cargo \cite{Graver} approached the problem from a
graph-theoretical perspective using so-called covering graphs.
All of these algorithms also compute the number of inequivalent immersions.
Another fact we glean from Blank is that any \SO\ curve $\gamma$
necessarily makes one full turn, i.e., it has Whitney index
$\whit{\gamma} = 1$.  The necessity of Whitney index 1 and positive
consistency to be \SO\, are well-known and date back to
\cite{TT61}.

\paragraph{Minimum Homotopy Area.}
The \myemph{minimum homotopy area} $\mathemph{\sigma(\gamma)}$ is
the infimum of areas swept out by nullhomotopies of a closed plane
curve $\gamma$. The key link between minimum area homotopies and \SO\ curves arose in
\cite{fkw_mahncp_17,Karakoc}, where the authors showed that any curve $\gamma$ has a
minimum area homotopy realized by a sequence of nullhomotopies of
\SO\ subcurves (direct split subcurves; see \secref{directsplits} for the definition).
%Hence, the topological/analytic problem of
%computing minimum homotopy area is reduced to one of discrete combinatorics.
The minimum homotopy area was introduced by Chambers and Wang
\cite{CW13} as a more robust metric for curve comparison than homotopy
width (i.e., Fr\'echet distance or one of its variants) or homotopy height
\cite{Chambers}.  The minimum homotopy area can be computed in
$O(N^2\log N)$ time for consistent curves \cite{CW13}. For general
curves, Nie gave an algorithm to compute $\sigma(\gamma)$ based on an
algebraic interpretation of the problem that runs in~$O(N^6)$ time,
while the \SO\ decomposition result of \cite{fkw_mahncp_17} yields an
exponential-time algorithm.
%
%While \SO\ curves clearly tell us about the structure of minimum area
%homotopies, our second result we show that minimum area homotopies can actually
%reveal to us facts about the combinatorial structure of \SO\ curves.
%
The \myemph{winding area} $\mathemph{W(\gamma)}$ is the integral over all
winding numbers in the plane. A simple argument shows that~$\sigma(\gamma) \geq
W(\gamma)$; see \cite{CW13}.  Both
\SO\ curves and interior boundaries are characterized by positive consistency and
optimality in minimum homotopy area, $\sigma(\gamma) = W(\gamma)$. A curve
possessing both of these properties is \SO\ when~$\whit{\gamma} = 1$
and an interior boundary when $\whit{\gamma} \geq 1$.

\subsection{New Results}
We ask
the following question: what are the sufficient combinatorial conditions
for a plane curve to be \SO?
 Such conditions provide novel mathematical foundations that
 could pave the way for speeding up algorithms for related problems, such as
 deciding \SO ness or computing the minimum homotopy area of a curve.
The first contribution of this paper is to answer this question in the
affirmative (Theorems \ref{thm:ISO} and \ref{thm:THSO} in \secref{wraps}):
We show that a curve $\gamma$ with Whitney index $1$
and without any \SO\ subcurves is \SO, and
%; we call such curves \myemph{irreducible}.
we obtain sufficient conditions for a
curve to be \SO\ solely in terms of the Whitney index of the curve and its subcurves.
%\footnote{We also assume $\gamma$ has a positive outer basepoint, a technical
%condition defined in \secref{regularGenericCurves}.}.
% \carola{Removed footnote to save space and to avoid Latex font warning.}
%
Here, we only consider \myemph{direct split} subcurves $\gamma_v$ that traverse $\gamma$ between the first and second appearance of vertex $v$ in the plane graph induced by $\gamma$.
Our results apply to (strongly) irreducible curves; see \figref{smallerCurveLattice}:
We call $\gamma$ \myemph{irreducible},
if every (proper) direct split is not \SO; if the Whitney index of each such
direct split is non-positive, then we call $\gamma$ \myemph{strongly irreducible}.

These results follow from our second contribution (Theorems
\ref{thm:WRB} and \ref{thm:WRB2} in \secref{wraps}), which shows that any plane curve
$\gamma$
%, modulo a basepoint condition,
is transformed into an
interior boundary by wrapping around $\gamma$ with Jordan curves.
Equivalently, this means that the minimum homotopy
area of~$\gamma$ is reduced to the minimal possible threshold, namely
the winding area, through wrapping.
See \figref{wrapIntro} for an example of wrapping.
Of course it is easy to see that---after repeated wrapping---a curve becomes positive consistent,
since a single wrap increases the winding numbers of each face by one. However, our result shows a new and non-trivial connection between wrapping and the minimum homotopy area.

\begin{figure}[hbt]
\centering
    \includegraphics[width=.5\textwidth]{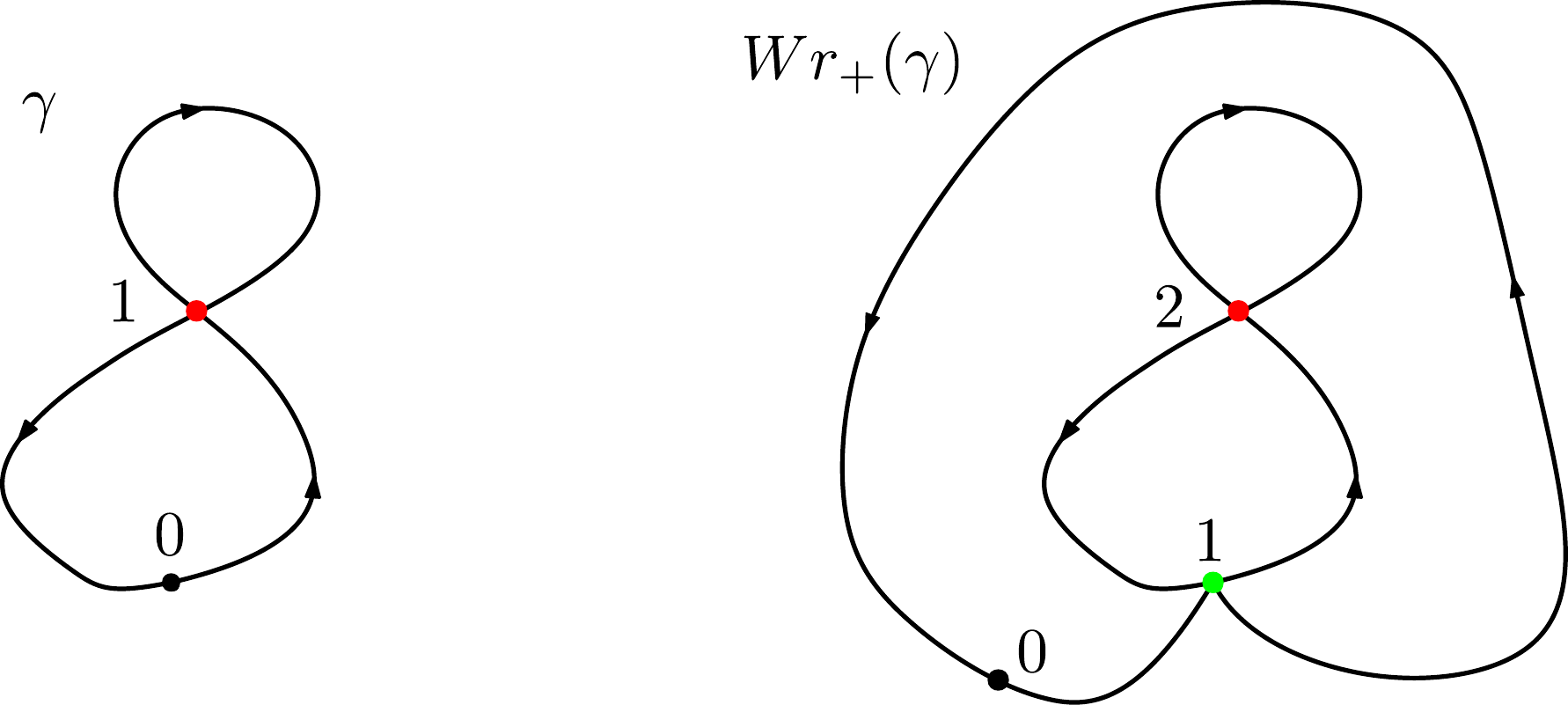}
    \caption{The curve $\gamma$ is not \SO, but its wrap $\Wr_+(\gamma)$ is \SO.
    }
   \label{fig:wrapIntro}
\end{figure}
%

%We introduce new curve classes, \myemph{irreducible}
%curves essentially do not have any \SO\ subcurves, and
%\myemph{strongly irreducible} curves

The third contribution of this paper (in \secref{SO_Int}) is to unite the various
definitions and perspectives on \SO\ curves and interior
boundaries. We prove the equivalence of five definitions of
\SO\ curves and four of interior boundaries (Theorems \ref{thm:ESO} and \ref{thm:EIB}).
To this end, we define the new concept of \myemph{obstinance} of a
curve $\gamma$ as $\mathemph{\obs(\gamma)} = \sigma(\gamma)-W(\gamma)
\geq 0$, and characterize \myemph{zero-obstinance} curves (\thmref{Ideal}), see \figref{smallerCurveLattice}.
Rephrasing our earlier characterization, \SO\ curves and interior
boundaries are positive-consistent curves with zero-obstinance and
positive Whitney index.

We conclude by defining a new operation called
\emph{balanced loop insertion}, a complementary notion to that of
\emph{balanced loop deletion}, the key trick to proving \thmref{ISO}.
As a parallel to our results on wraps, we show in
\thmref{GBLI} that careful iteration of balanced loop insertion turns
any curve $\gamma$ with $\whit{\gamma} = 1$ (and positive outer
basepoint) into a \SO \, curve.

More supplementary details on the relationship between different curve classes studied in this paper are provided in \appendixref{lattice}.%, including interior boundaries and consistent curves, are given in .
% See \thmref{CurveClasses} and \figref{CurveClasses} in 

%
%\figref{smallerCurveLattice} gives an overview of the curve classes considered in this paper.

%% Our contributions provide new mathematical insight into
%% self-overlapping curves and interior boundaries that could pave the
%% way for speeding up algorithms for related problems such as deciding
%% \SO ness or computing the minimum homotopy area of a curve.

%Overall, we introduce three new curve classes (strongly irreducible, irreducible, and zero-obstinance)

%% The remainder of this paper is organized as follows:
%% After some preliminaries (\secref{preliminaries}), we prove the equivalences between the different
%% definitions of \SO\ curves (\thmref{ESO}) and of interior boundaries
%% (\thmref{EIB}), and we introduce zero-obstinance curves (\secref{SO_Int}).
%% In \secref{wraps} we then rigorously define the wrapping operation and
%% prove our results on wraps (\thmref{WRB} and
%% \thmref{WRB2}) using the equivalences of interior boundaries. Next we
%% prove our new combinatorial conditions for plane curves to be \SO, \thmref{ISO} and \thmref{THSO}, using \thmref{WRB2}.
%

\section{Preliminaries} \label{sec:preliminaries}
%We now lay the necessary groundwork on planar curves, homotopies, self-overlapping
%curves, interior boundaries, and minimum area homotopies that are needed for this paper.

\subsection{Regular and Generic Curves}\label{sec:regularGenericCurves}

We work with regular, generic, closed plane curves
$\gamma: [0,1] \rightarrow \mathbb{R}^2$ with basepoint
$\gamma(0)=\gamma(1)$.  Let $\mathemph{\C}$ denote the set of such curves.
%We prefer to use $[0,1]$ instead of $\mathbb{S}^1$ to employ its well-ordering when defining combinatorial relations in Section 2.2.
A curve is \myemph{regular} if~$\gamma'(t)$ exists and
is non-zero for all $t$; a curve is \myemph{generic} (or
normal) if the embedding has only a finite number of
intersection points, each of which are transverse crossings.
Being generic is a weak restriction, as
normal curves are dense in the space of regular curves~\cite{WH37}.
Viewing a generic curve $\gamma$ by its image $ \mathemph{[\gamma]} \subseteq
\mathbb{R}^2 $, we can treat $\gamma$ as a directed plane
multigraph $\mathemph{G(\gamma)}=(V(\gamma),E(\gamma))$.
Here, $\mathemph{V(\gamma)}=\{p_0, p_1, \ldots, p_n\}$ is the set of ordered \myemph{vertices} (points)
of $\gamma$, with basepoint $p_0=\gamma(0)$ regarded as a vertex as well.
An \myemph{edge} $(p_i, p_j)$ corresponds to a simple path along $\gamma$
between $p_i$ and $p_j$. The \myemph{faces} of $G(\gamma)$
are the maximally connected components of $\R^2 \setminus[\gamma]$. Each $\gamma \in \C$ has exactly one unbounded face,
the exterior face $\mathemph{F_{ext}}$. See \figref{SO_curve}.
Two curves are combinatorially equivalent when their planar multigraphs are isomorphic.
We may therefore define a curve just by its image, orientation, and basepoint.
A curve is \myemph{simple} if it has no intersection points. We notate
\mathemphformula{$\abs{\gamma}$}$=\abs{V(\gamma) \setminus \{ p_0\}}$
as the complexity of $\gamma$.

For any $x \in \R^2 \setminus [\gamma]$, the \myemph{winding number}
$\mathemph{wn(x,\gamma)} = \sum_{i} a_i$ is defined using a simple path $P$ from $x$
to $F_{ext}$ that avoids the intersection points of $ \gamma$. Here,
$a_i=+1$ if $P$ crosses $\gamma$ from left to right at the $i$-th
intersection of $P$ with $\gamma$, and $a_i=-1$ otherwise.
%% We define the \myemph{winding number} of $x \in \R^2 \setminus
%% [\gamma]$ with respect to $\gamma$ using a simple path $P$ from $x$ to
%% $F_{ext}$ that avoids the intersection points of $ \gamma$. Suppose
%% the path $P$ crosses $\gamma$ at $\{y_i\}_{i=1}^j$.  Then we define
%% $\mathemph{wn(x,\gamma)} = \sum_{i=1}^j a_i$ where $a_i = +1$ if $P$
%% crosses $\gamma$ from left to right at $y_i$ and $a_i = -1$ otherwise.
Since this number is independent of the path chosen and is constant over
each face~$F$ of of $G(\gamma)$, we
write $\mathemph{wn(F,\gamma)}$.  
 If $wn(F, \gamma) \geq 0$ for every face
$F$ on $G(\gamma)$, then we call $\gamma$ \myemph{positive
  consistent}.  If $wn(F, \gamma) \leq 0$ for every face, then $\gamma$
is \myemph{negative consistent}.  See \figref{SO_curve} for an example curve illustrating these concepts.
The \myemph{winding area} of a curve~$ \gamma$ is given~by
$\mathemph{W(\gamma)} = \int_{\mathbb{R}^2} \abs{wn(x, \gamma)} dx = \sum_{F } A(F) \abs{wn(F, \gamma)}$,
where $A(F)$ is the area of the face $F$
and $wn(x, \gamma) = 0$ for $x \in [\gamma]$.
The \myemph{depth} $D(F, \gamma)$ of a face $F$ is the minimum
number of crossings of any simple path $P$ from $F$ to $F_{ext}$
with $\gamma$.
The \myemph{depth} of $\gamma$ is the sum $D(\gamma) =
\sum_{F} A(F) D(F, \gamma)$.
The \myemph{Whitney index} $\mathemph{\whitbold{\gamma}}$ of a curve~$\gamma$ is
the winding number of the derivative~$\gamma'$ about the
origin.
A curve $\gamma$ is \myemph{positively oriented} if $\whit{\gamma}>0$ and \myemph{negatively oriented} if $\whit{\gamma}<0$.

\begin{figure}[ht]
    \centering
    \includegraphics[width=.4\textwidth]{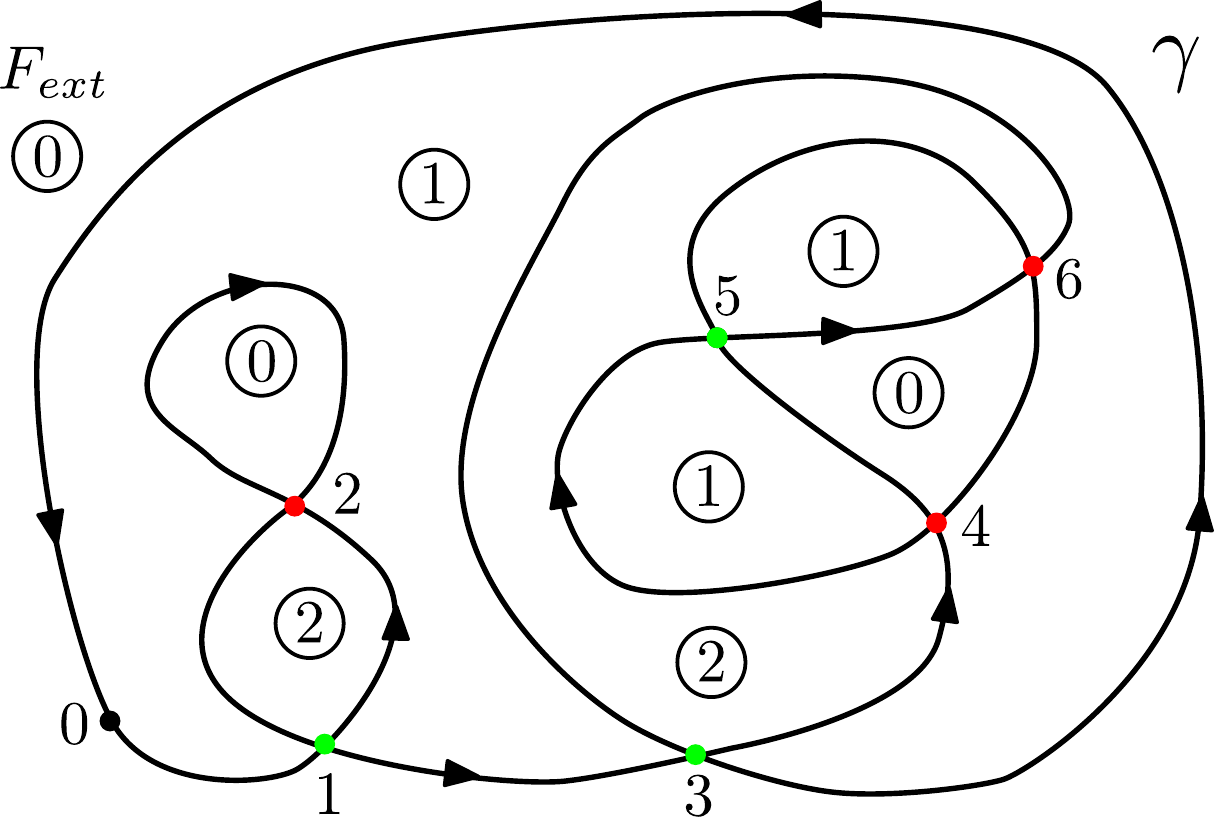}
    \caption{A curve $\gamma$ that is \SO.  The winding numbers of each
    face are enclosed by circles. The signed
    intersection sequence of $\gamma$ is $0, 1_+, 2_-, 2_+, 1_-, 3_+, 4_-, 5+_, 6_-, 4_+, 5_-, 6_+, 3_-, 0$; vertex labels are shown, and the sign of each vertex is indicated with green (positive) or red (negative). The combinatorial relations are: $p_2\subset p_1$; $p_4,p_5,p_6\subset p_3$; $p_1,p_2 \separate p_3,p_4,p_5,p_6$; $p_4,p_5 \links p_6$; $p_4 \links p_5$.
   }
   \label{fig:SO_curve}
\end{figure}

%Positive outer basepoint
%We call $p \in [\gamma]$ \myemph{outer} if it lies on the boundary of
%$F_{ext}$.
%, and \myemph{inner} otherwise.
A basepoint $p_0= \gamma(0)$ is a \myemph{positive outer basepoint} if
$p_0$ is incident to the two faces $F_{ext}$ and $F$, and $wn(F,\gamma)=1$.
If $wn(F,\gamma)=-1$, then $p_0$ is a negative outer basepoint.
Several of our results require $\gamma$ to have a positive outer basepoint.
A curve $\gamma: \mathbb{S}^1 \rightarrow \mathbb{R}^2$ is \myemph{(positive) \SO} when there is
an orientation-preserving immersion $F: \mathbb{D}^2 \rightarrow \mathbb{R}^2$, a map of full rank,
extending $\gamma$ to a map on the entire two-dimensional unit disk
$\mathbb{D}^2$. If the reversal $\overline{\gamma}$ of
a curve is \SO, then we call $\gamma$ \myemph{negative \SO}.  Unless stated otherwise, the term \SO\ is used
only to mean positive \SO.
%
%\parker{Need to ensure that we no longer call curves just `positive' or `negative',
%and instead spell out positively (resp.) negatively oriented. This avoids us
%presenting the reader with another unnecessary definition/convention.}

\subsection{Combinatorial Relations
  and Intersection Sequences}

Following Titus \cite{TitusNC}, we now describe how the intersection
points of a curve $\gamma\in\C$ relate to each other.
See \figref{SO_curve} for an illustration of these concepts.
%
%\begin{definition}[Combinatorial Relations \& Intersection Sequences]
  Let
  $p_i, p_j\in V(\gamma)$ be two vertices such that
  $p_i=\gamma(t_i)=\gamma(t_i^*)$ and $p_j=\gamma(t_j)=\gamma(t_j^*)$
  with $t_i < t_i^*$ and $t_j < t_j^* $.
  Then, one of the following must~hold:
  %
  %  \begin{enumerate}[align=left,labelindent=0pt,labelwidth=\widthof{\ref{last-item}},label=(\roman*),itemindent=1em,leftmargin=!]
\begin{itemize}
\item $p_i$ \myemph{links} $p_j$, or $p_i\mathemphtxt{\links}p_j$, iff $t_i < t_j < t_i^* < t_j^*$ or $t_j < t_i < t_j^* < t_i^*$
\item $p_i$ is \myemph{separate} from $p_j$, or $p_i\mathemphtxt{\separate}p_j$, iff  $t_i < t_i^* < t_j < t_j^*$ or $t_j < t_j^* < t_i < t_i^*$
\item $p_i$ is \myemph{contained in} $p_j$, or $p_i~\mathemph{\subset}~p_j$, iff  $t_j \leq t_i < t_i^* \leq t_j^*$
\end{itemize}
%
%Any pair $p_i, p_j \in V(\gamma)$ must have one of these
%relations.
%See Figure \ref{SO_curve}.
%
To define the intersection sequence of $\gamma$,
the vertices are labeled in the order they appear on $\gamma$,
starting with $0$ for the basepoint $\gamma(0)$, and increasing
by one each time an unlabeled vertex is encountered. The
\myemph{signed intersection sequence} consists of the sequence of all
vertex labels along $\gamma$ starting at the basepoint; the first time
vertex $p_i$ is visited, the label is augmented with $\sgn(p_i)$, and
the second time with $-\sgn(p_i)$. Here,
${\sgn(p_i)} :=\sgn(p_i,\gamma)$ is the  \myemph{sign} of vertex
          $p_i=\gamma(t_i)=\gamma(t_i^*)$, and
is $1$ if the vector~$\gamma'$ rotates clockwise from $t_i$ to $t_i^*$, and
$-1$ otherwise. Note that $\sgn(p_i)$ depends on the basepoint of the
curve. 
%
%% \footnote{We prefer to begin and end the intersection
%%   sequence with $0$, corresponding to the basepoint, so that each
%%   consecutive pair $i,j$ in the intersection sequence
%%   corresponds to a unique edge from $p_i$ to $p_j$.}
%The signed intersection sequence carries equivalent information to the combinatorial relations.
%
As proved by Titus, interior boundariness is invariant with respect to
signed intersection sequences ~\cite{TT61}.
%We shall take this fact for granted throughout the paper. 

\subsection{Minimum Homotopies}
%\subsection{Homotopies and Homotopy Moves}

A \myemph{homotopy} between two generic curves $\gamma$ and
$\gamma'$ is a continuous function $H: [0,1]^2 \rightarrow
\mathbb{R}^2$ such that $H(0, \cdot)= \gamma$ and $H(1,\cdot) =
\gamma' $. In $\mathbb{R}^2 $, any curve is null-homotopic, i.e.,
homotopic to a constant map.  Given a sequence of homotopies
$(H_i)_{i=1}^k$, we notate the concatenation of these homotopies in
order as \mathemphformula{$\sum_{i=1}^k H_i$}. We use the notation $
\overline{H}$ for the reversal $\overline{H}(i,t) = H(1-i, t)$ of a
homotopy. If $H(0, \cdot) = \gamma$ and $H(1, \cdot) = \gamma'$, we
may write \mathemphformula{$\gamma \homotopyarrow{H} \gamma'$}.

\myemph{Homotopy moves} are basic local alterations to a
curve defined by their action on~$G(\gamma)$. These moves come in
three pairs ~\cite{fkw_mahncp_17}, see~\figref{SimpleTitusMoves}: The I-moves destroy/create
an empty loop, II-moves destory/create a bigon, and III-moves flip a
triangle.
%For I- and II-moves,
We denote the moves that remove
vertices as $\mathemphtxt{\text{I}}_{\mathemph{a}}$ and $\mathemphtxt{\text{II}}_{\mathemph{a}}$, and moves that
create vertices as $\mathemphtxt{\text{I}}_{\mathemph{b}}$ and $\mathemphtxt{\text{II}}_{\mathemph{b}}$.
See \figref{SimpleTitusMoves}.
It is well-known that any homotopy such that each intermediate curve
is piecewise regular and generic, or almost generic,
can be achieved by a sequence of homotopy moves.
Thus, without loss of generality, we assume that each time
the curve $H(i, \cdot)$ combinatorially changes is through a single
homotopy move.

\begin{figure}[htbp]
    \centering
    \includegraphics{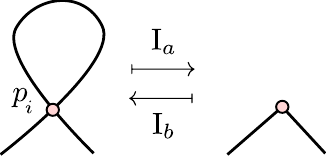}
    \hspace*{1ex}
    \includegraphics{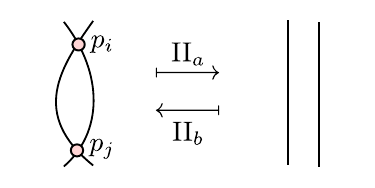}
    \hspace*{1ex}
    \includegraphics{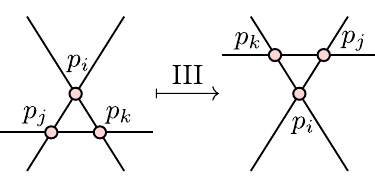}
    \caption{All three homotopy moves and their reversals.
      %As shown, $a$-moves destroy
    %vertices, $b$-moves create vertices, and III-moves do not create or destroy vertices.
    Figure from \cite{fkw_mahncp_17}.}
    \label{fig:SimpleTitusMoves}
\end{figure}

%We define the minimum homotopy area, as in \cite{CW13,fkw_mahncp_17}:
%
%\begin{definition}[Minimum Homotopy Area]
%
  Let $\gamma \in \C$ and $H$ be a nullhomopty of $\gamma$.
Define~$E_H(x)$ as the number of connected components of
$H^{-1}(x)$. Intuitively, this counts the number of times that $H$ sweeps over $x$.
The \myemph{minimum homotopy area} of $\gamma$ is defined as
$\mathemph{\sigma(\gamma)} = \inf_H \{ \int_{\mathbb{R}^2} E_H(x) \text{ } dx\;|\; H \;\text{is a}$ $\text{nullhomotopy of} \; \gamma\}$.
%\end{definition}
%
%A basic, but essential, property of the minimum homotopy area is that it is bounded
%below by the winding area, as shown in \cite{CW13,fkw_mahncp_17}:
The following was shown in \cite{CW13,fkw_mahncp_17}:

%The minimum homotopy area is bounded from
%below by the winding area, as shown in \cite{CW13,fkw_mahncp_17}:

\begin{lemma}[Homotopy Area $\geq$ Winding Area]\label{lem:minHomotopy_and_windingAreas}
Let $\gamma \in \C$. Then $\sigma(\gamma) \geq W(\gamma)$.
\end{lemma}

A straightforward proof by induction, similar to that of Lemma \ref{lem:minHomotopy_and_windingAreas} shows the following.

\begin{lemma}[Homotopy Area $\leq$ Depth]\label{lem:minHomotopy_and_depth}
Let $\gamma \in \C$. Then $\sigma(\gamma) \leq D(\gamma)$. 
\end{lemma}

On the directed multigraph $G(\gamma)$, we can  define
the left and right face of any edge.
We call a homotopy \myemph{left (right) sense-preserving} if $H(i + \epsilon, t)$ lies on or
to the left (right) of  the oriented curve $H(i,\cdot)$
%at $H(i, t)$
for any $i,t \in [0,1]$ and any $\epsilon > 0$.
%\myemph{Right sense-preserving} homotopies are defined analogously.
%
The following two lemmas provide useful properties about sense-preserving homotopies; the first was proven in \cite{CW13}, the second in \cite{fkw_mahncp_17}.
\begin{lemma}[Monotonicity of Winding Numbers]\label{lem:sense-preserving-monotone}
  Let $H$ be a homotopy. If $H$ is left (right) sense-preserving, then for any $x\in\R^2$ the function $a(i) = wn(x, H(i, \cdot))$ is monotonically decreasing (increasing).
  %If $H$ is right sense-preserving, then $a(i)$ is monotonically increasing.
\end{lemma}

\begin{lemma}[Sense-Preserving Homotopies are Optimal]\label{lem:sensePreservingOptimal}
%  Let $\gamma\in\C$. Then $obs(\gamma)=0$ if and only if any nullhomotopy of $\gamma$ with optimal homotopy area is sense-preserving.
  Let $\gamma\in C$ be consistent. Then a nullhomotopy H of $\gamma$ is optimal if and only if it is sense-preserving. 
\end{lemma}

\section{Equivalences}\label{sec:SO_Int}
In this section, we show the equivalence of different characterizations of interior boundaries
(\thmref{EIB}) and of \SO\ curves (\thmref{ESO}). Our analysis of curve classes hinges around
the concept of obstinance.
 In \thmref{Ideal} we classify zero obstinance curves, which are
generalizations of interior boundaries and of \SO\ curves.
%Examining the interconnected lattice of curve classes \figref{CurveClasses},
%it is evident that obstinance
%reveals many interesting connections between different classes of curves.

\subsection{Direct Splits}\label{sec:directsplits}

%% The notion of a free subcurve is implicit in \cite{fkw_mahncp_17}.
%% %
%% However, not all free subcurves are created equal, and our distinction
%% between direct splits and indirect splits will be instrumental throughout
%% this paper.

%\begin{definition}[Direct Splits]\label{splits}
Let $\gamma \in \script$ and $p_i \in V(\gamma)$
with $p_i=\gamma(t_i) = \gamma(t_i^*)$ and $t_i < t_i^*$.
Then, $\gamma$ can be split into two subcurves at $p_i$:
The \myemph{direct split}
%of $\gamma$ with respect to $p_i$
is the curve with image $[\gamma|_{[t_i, t_i^*]} ]$ with basepoint $p_i$,
and the \myemph{indirect split}
%of $\gamma$ with respect to $p_i$
is the curve with image
$[\gamma|_{[t_i^*, 1]}] \cup [\gamma|_{[0, t_i]}]$ with basepoint $\gamma(0)$.
We endow both of these curves with the same orientation as $\gamma$.
Given a direct (or indirect) split $\widetilde{\gamma}$ on a curve
$\gamma$, we write $\gamma \setminus \widetilde{\gamma}$ for the indirect (or direct) split complementary to $\widetilde{\gamma}$.
%Since the basepoint $p_0 \in V(\gamma)$, note that
%the whole curve, $\gamma$, is a direct split of itself. Hence,
We call 
a direct split \myemph{proper} if it is not the entire curve~$\gamma$.
See \figref{decomps}.
If $v=p_i\in V(\gamma)$, we may notate the direct split as $\gamma_i$ or $\gamma_v$.
%, and the indirect split as~$\gamma_{i^*}$ or $\gamma_{v^*}$.
%Note that the basepoint
%of a indirect split $\gamma_{v^*}$ is the same as the basepoint of the original curve $\gamma$,
%meanwhile, the basepoint of a direct split $\gamma_v$ is $v$.
%
When removing multiple splits iteratively
we write~$\gamma \setminus \left(\cup_{i=1}^n \gamma_i \right)$, where
we require that $\gamma_i$ is a direct split of $\gamma \setminus \left(\cup_{j=1}^{i-1} \gamma_j \right)$.
The direct splits carry a great deal of information about the curve. In fact,
one can recover the combinatorial relations from the direct splits:
$p_i \links p_j$ iff~$p_i \in [\gamma_j]$ and~$p_j \in [\gamma_i]$;
$p_i \separate p_j$ iff $p_i \notin [\gamma_{j}]$
and~$p_j \notin [\gamma_{i}]$;
and $p_i \subset p_j$ iff $p_i$ is a vertex of $\gamma_{j }.$
Being a direct split of a curve is a transitive property. I.e.,
if~$\gamma_i \in \script$ is a direct split on~$\gamma$, and $\gamma_j$ is a direct split on $\gamma_i$, then
$ \gamma_j$ is a direct split on~$\gamma$.
The parallel statement on indirect splits, however, is false.

\begin{figure*}[htbp]
    \centering
    \includegraphics[width=.70\textwidth]{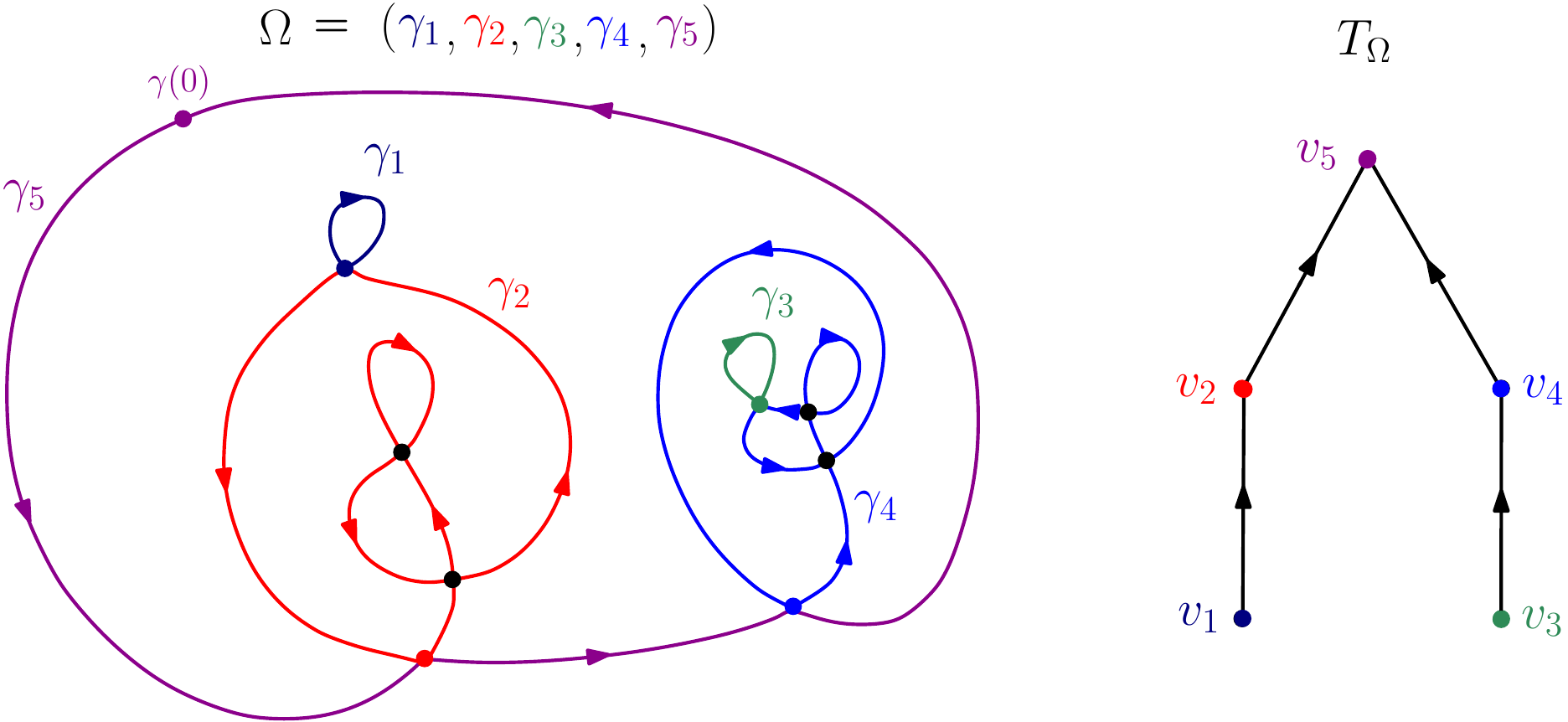}
    \caption{A \SO\ decomposition of a \SO\ curve $\gamma$.
Here,
    $\gamma_1$ and $\gamma_3$ are (proper) direct splits of $\gamma$, while
    $\gamma_2, \gamma_4,$ and $\gamma_5$ are neither direct nor
    indirect splits of $\gamma$.}
%Note    that $\gamma_1$ and $\gamma_3$ are the only direct splits in the decomposition. The subcurves $\gamma_2, \gamma_4,$ and $\gamma_5$ are not free.}
    \label{fig:decomps}
\end{figure*}

\subsection{Decompositions and Loops}

A curve $\gamma \in \C$ can be
entirely decomposed by iteratively removing direct splits.
%
%We call a sequence of such subcurves  a
%\myemph{free subcurve decomposition}.  Precisely,
For a sequence of subcurves $\Omega = (\gamma_i)_{i=1}^k$,
define $C_0 = \gamma$ and inductively
$C_i = C_{i-1} \setminus \gamma_i$ for~$i\geq 1$;
the basepoint of $\gamma_i$ is $v_i=C_i\cap\gamma_i$.
We call $\Omega$ a \myemph{direct split
decomposition} if $\gamma_i$ is a direct split of $C_{i-1}$, for all
$i \in \{1, 2, \dots, k\}$, and $\gamma_{k} = C_{k-1}$.
Note that
$\Omega$ nearly induces a partition of $\gamma$ in the sense that
$[\gamma] = \cup_{i=1}^k [\gamma_i]$ and 
$\gamma_i \cap \gamma_j \subset V(\gamma)$ for any $i\neq j$.
Given a direct split decomposition $\Omega = (\gamma_i)_{i=1}^k$, we
write~$\mathemph{V(\Omega)}$ for the set of basepoints of all
$\gamma_i \in \Omega$. See Figure \ref{fig:decomps}. 

Observe that no two vertices $v_i, v_j\in V(\Omega)$ may be linked.
Hence, we obtain a partial order~$\prec$ on~$V(\Omega)$ by declaring
$v_i \prec v_j$ whenever $v_i \subset v_j$.
We define $T_{\Omega}$ to be the rooted, directed tree with vertex
set~$V(T_{\Omega}) = V(\Omega)$ and edges $e=(v_i, v_j)$ whenever
$v_i \subset v_j$ and there is no other vertex~$v_k \neq v_i, v_j$
such that~$v_i \subset v_k \subset v_j$.
The root of $T_{\Omega}$ corresponds to the basepoint of $\gamma$.
We consider two subcurve decompositions~$\Omega, \Gamma$
equivalent,~$\Omega \sim \Gamma$, when $T_{\Omega} = T_{\Gamma}$.
One can easily verify that $\sim$ is
an equivalence relation on the set of direct split decompositions of
$\gamma$.
This means that $\Omega$ and $\Gamma$
contain the same set of subcurves, just in a different order.
If every $\gamma_i$ is \SO, we call $\Omega$
a \myemph{\SO\ decomposition};
it may contain \SO\ subcurves of positive and negative orientations.
We now observe that the vertex set of a
decomposition already determines the subcurves in the decomposition:

\begin{observation}
Given a curve $\gamma \in \C$ and a subset $S \subset V(\gamma)$ such
that $p_0 \in S$ and no two vertices in $S$ are linked, there is a
unique equivalence class $\mathscr{E}$
%= \{\Omega_i\}_{i=1}^m$
of direct split decompositions with $V(\Omega) = S$ for all $\Omega\in\mathscr{E}$.
\end{observation}

%We write $\mathemph{W(\Omega)} = \sum_{C \in \Omega} W(C)$
%for the winding area of the decomposition. 
%
The observation below follows directly from the definition of winding numbers.

\begin{observation}\label{obs:wnDecomp}
Let $\Omega$ be a direct split decomposition of a curve $\gamma \in \C$. Then for any face $F$ in the plane multigraph $G(\gamma)$, $wn(F, \gamma) = \sum_{\gamma_i \in \Omega} wn(F, \gamma_i)$.
\end{observation}

We define a \myemph{loop} as a simple direct split $\gamma_v$ of a curve $\gamma\in\C$.
Intersection points of~$\gamma$ may lie on $\gamma_v$, but none occur
as intersections of $\gamma_v$ with itself.  Every non-simple plane
curve has a loop; e.g., the direct split $\gamma_w$, where $w$ is
the highest index vertex on $\gamma$ in the signed intersection sequence.
A loop $\gamma_v$ is \myemph{empty} if $v$ links no vertex $w \in
V(\gamma)$. Let $\text{int}(\gamma_v)$ denote its interior.  We call
$\gamma_v$ an \myemph{outwards loop} if the edges $e_1,e_4$, that are
incident on $v$ and lie on $\gamma\setminus\gamma_v$,
%=\gamma_{v^*}$,
both lie outside $\text{int}(\gamma_v)$.  Otherwise
$\gamma_v$ is an \myemph{inwards loop}.  See \figref{loopEx}.

\begin{figure}[htbp]
\center
\includegraphics[height=2.2cm]{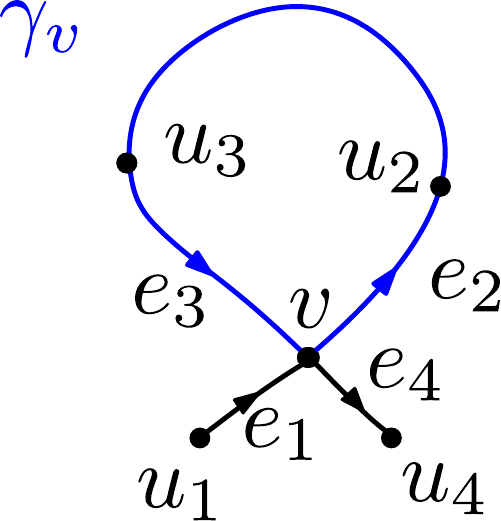}
\hspace*{2cm}
\includegraphics[height=2.1cm]{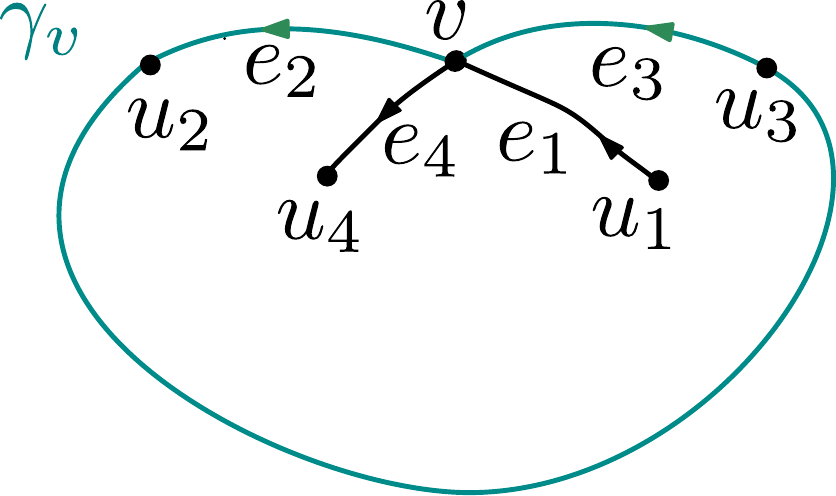}
\caption{An outwards loop (left) and an inwards loop (right).}
\label{fig:loopEx}
\end{figure}

The lemma below follows from \cite{Gauss,Seifert}. Since it requires a digression from our 
main focus, its proof is given in \appendixref{whit}.

\begin{restatable}[Whitney Index Through Decompositions]{lemma}{WHSSD}\label{lem:WHSSD}
 Let $\gamma \in \C$ and $\Omega$ be a direct split decomposition of $\gamma$.
 Then $\whit{\gamma} = \sum_{C \in \Omega} \whit{C}$.
 \end{restatable}

A consequence of \lemref{WHSSD} is that iteratively removing loops and summing $\pm 1$ for their
signs allows one to quickly compute Whitney indices.  
Assuming $\gamma$ is given as a directed plane multigraph,
one can adapt a depth-first traversal to compute such a \decomp\ of $\gamma$ in $O(|\gamma|)$ time, which yields the following corollary:
%
%\carola{Add description of algo here or to appendix?}

\begin{corollary}[Compute Whitney Index]\label{cor:computeWI}
Let $\gamma \in \C$ be of complexity $n=|\gamma|=|V(\gamma)|$. One can compute a \decomp\ of $\gamma$, and $\whit{\gamma}$, in $O(n)$ time.
\end{corollary}

\subsection{Well-Behaved Homotopies}

Let $H$ be a nullhomotopy of a curve $\gamma$, and consider all the
points $A = \{v_i\}_{i=1}^k$ of~$\R^2$ such that $H$ performs a
$\text{I}_a$ move to contract a loop to that point.
All such points are called \myemph{anchor points} of the homotopy $H$.
Following \cite{fkw_mahncp_17} we call a homotopy $H$
\myemph{well-behaved} when the anchor points $A$ of $H$ satisfy
$A \subseteq V(\gamma)$, i.e., $H$ only contracts loops to vertices of the original curve, not to new vertices created along the way by $H$.
The theorem below from \cite{fkw_mahncp_17} shows that computing minimum homotopy area is reduced to
finding an optimal \SO\ decomposition.
The homotopy
$H$ guaranteed in the following theorem is well-behaved.

\begin{theorem}\label{thm:mhd} [Minimum Homotopy Decompositions]
Let $\gamma \in \C $. Then there is a \SO\ decomposition $\Omega = (\gamma_i)_{i=1}^k$ of $\gamma$
as well as an associated minimum homotopy $H_{\Omega}$ of $\gamma$ such that
$H_{\Omega} = \sum_{i=1}^k H_i$ and each $H_i$ is a nullhomotopy
of $\gamma_i$.
%In particular,~$\sigma(\gamma) = \min_{\Omega \in \mathscr{D}(\gamma)} W(\Omega),$
In particular,~$\sigma(\gamma) = \min_{\Omega \in \mathscr{D}(\gamma)} \sum_{C \in \Omega} W(C)$,
where $\mathscr{D}(\gamma)$ is the set of all \SO\ decompositions of $\gamma$.
\end{theorem}

\subsection{Equivalence of Interior Boundaries}

In this section, we unify different definitions and
characterizations of interior boundaries by showing their
equivalence.
%
%We prefer to begin with interior boundaries, so that \SO\
%curves can just be thought of as $1$-interior boundaries in the following
%subsection.
%
We call a curve $\gamma$ a \myemph{$\mathbf{k}$-interior boundary} when
(1)~$\obs(\gamma) = 0$,
(2)~$\whit{\gamma} = k>0$, and
(3)~$\gamma$ is positive consistent.
We call~$\gamma$ a \myemph{$(\mathbf{-k})$-interior boundary} when its reversal $\overline{\gamma}$ is a $k$-interior boundary. 
In accordance with Titus~\cite{TT61}, we call a curve $\zeta: [0,1] \rightarrow \mathbb{R}^2$
a \myemph{Titus interior boundary}   if there exists a map $F:
\mathbb{D}^2 \rightarrow \mathbb{R}^2$ such that $F$ is continuous, \myemph{light}
(defined as: pre-images are totally disconnected), open, orientation-preserving, and
$F|_{\partial \mathbb{D}^2} = \zeta $. The map $F$ is called \myemph{properly
interior}.

We prove the equivalence of these definitions and of two further characterizations below. 
%The interested reader may note that \propref{EIB:iii} of \thmref{EIB} was the definition of an interior boundary in \cite{Karakoc}.

\begin{theorem}[Equivalence of Interior Boundaries]\label{thm:EIB}
    Let $\gamma \in \C$ have a positive outer basepoint~$\gamma(0)$,
    and suppose $\whit{\gamma} = k>0$. Then, the following are equivalent:
    \begin{enumerate}
        \item $\gamma$ is an interior boundary.\label{EIB:i}%(i)
        \item $\gamma$ is a Titus interior boundary.\label{EIB:ii}%(ii)
        \item $\gamma$ admits a \SO\ decomposition $\Omega = (\gamma_i)_{i=1}^k$, where
            each $\gamma_i$ is positive \SO. \label{EIB:iii}%(iii)
        \item $\gamma$ admits a well-behaved left sense-preserving nullhomotopy $H$ with exactly $k$
            $\text{I}_a$-moves. \label{EIB:iv}%(iv)
    \end{enumerate}
\end{theorem}

\begin{proof}
%We first prove \ref{EIB:i} $\Rightarrow$ \ref{EIB:ii}.
%
{\bf \ref{EIB:i} $\Rightarrow$ \ref{EIB:ii}}: We proceed by induction on $k$.
    If $k=1$, let $\gamma_1=\gamma$. The proof of this follows from
    \cite{TT61}, as well as \cite{BL67}.
%
    % Inductive step:
    Let $k \geq 1$ and let our inductive assumption be that for all $c \in \{1,
    2, \ldots, k\}$,
    if
    $\gamma \in \C$,
    %has a positive outer basepoint,
    $\whit{\gamma}=+c$, and $\gamma$ is an interior
    boundary, then $\gamma$ is a Titus interior boundary.
    % Inductive step:
    Now, let $\gamma \in \C$ such that
    %$\gamma$  has a positive outer basepoint,
    $\whit{\gamma}=+(k+1)$, and $\gamma$ is an interior
    boundary.  Since $\whit{\gamma} \geq 2$, we can find a vertex~$v$ of $\gamma$
    such that the two direct splits at $v$, $\gamma_1$ and $\gamma_2$,
    have positive Whitney indices~$k_1$ and~$k_2$, respectively.  Necessarily,
    $k_1 \leq k$ and $k_2 \leq k$, by Lemma \ref{lem:WHSSD}.
    By our inductive assumption, there exist properly interior maps $F_1,F_2
    \colon D^2 \to \R^2$ such that $F_1|_{\partial D^2} = \gamma_2$ and
    $F_1|_{\partial \mathbb{D}^2} = \gamma_2$.
    Finally, glue these two curves together at $v$ by finding an arc
    interior to both. (We refer the
    reader to Titus' paper \cite{TT61} to see a comprehensive explanation of his
    trick of gluing two properly interior mappings together along an interior
    arc.)
    The resulting map $F \colon \mathbb{D}^2 \# \mathbb{D}^2 \to \R^2$, where $\#$ denotes the
    connected sum,  extends both $F_1$ and $F_2$ and
    represents the curve $\gamma$.  Moreover,~$F$ is a properly interior map.
    We conclude that \ref{EIB:i} $\Rightarrow$ \ref{EIB:ii}.

{\bf \ref{EIB:ii} $\Rightarrow$ \ref{EIB:i}:} Let $\gamma$ be a Titus interior boundary.
    By \cite[Lemma 1]{TitusNC}, we know that $\gamma$ is consistent.
    Furthermore, by
    \cite{TT61}, we have $\abs{\gamma^{-1}(x)} = wn(x, \gamma)$ for all $x$.
    Thus, the linear
    retraction of the disk $D^2$ induces a homotopy $H$ with $A(H) = W(\gamma)$.
    By \lemref{minHomotopy_and_windingAreas}, we have~$\sigma(\zeta) = W(\gamma)$.
    Thus, $\gamma$ is an interior
    boundary, and so \ref{EIB:ii} $\Rightarrow$ \ref{EIB:i}.

{\bf \ref{EIB:i} $\Rightarrow$ \ref{EIB:iii}:} Let $\gamma$ be an
interior boundary. By \thmref{mhd}, we have an optimal \SO\
decomposition $\Omega = (\gamma_i)_{i=1}^j$ of $\gamma$. Suppose, by
contradiction, that there exists an~$l\leq j$ such that $\gamma_l$ is negative \SO.
%
%By \obsref{wnDecomp}, for any free subcurve decomposition
%and any face $F$ of $G(\gamma)$, we have $wn(F, \gamma) = \sum_{i=1}^k
%wn(F, \gamma_i) $.
%
%It follows that
Let $F$ be any face contained in the interior $\intt(\gamma_l)$. We know by \obsref{wnDecomp} that
$wn(F,\gamma) = \sum_{i=1}^j wn(F,\gamma_i)$, and since~$\gamma$ is positive consistent $wn(F,\gamma)\geq 0$.
Thus there must exist a positive \SO\ curve~$\gamma_i \in \Omega$ with~$F \subseteq \intt(\gamma_i)$.
%\carola{We don't seem to be using $\gamma_i$ in the proof.}
%
Consider the nullhomotopies $H_l$ and $H_i$ that are part of the canonical optimal homotopy $H_{\Omega}$.
Then~$H_l$ contracts $\gamma_l$ and is right sense-preserving, while~$H_i$
contracts~$H_i$ and is left sense-preserving. Thus by
\lemref{sense-preserving-monotone}, $H_l$ increases the winding number on~$F$
and $H_i$ decreases the winding number, which means $F$ is swept more than~$W(F)$ times, a contradiction.
    Thus, no negative \SO\ subcurve~$\gamma_l$ may
exist in~$\Omega$.
Since~$\whit{C} = 1$ for any positive \SO\ subcurve and~$\whit{\gamma} = \sum_{i=1}^k \whit{\gamma_i}$ by \lemref{WHSSD}, we we must have $k=j$.

{\bf \ref{EIB:iii} $\Rightarrow$ \ref{EIB:i}:} 
    Suppose $\gamma$ has a decomposition $\Omega$ into $k$ positive
    \SO\ subcurves. Then, the canonical homotopy $H_{\Omega}$
associated to $\Omega$ is left sense-preserving, since each minimum nullhomotopy of the subcurves is left sense-preserving. Since sense-preserving homotopies are optimal, see \lemref{sensePreservingOptimal}, we have $\sigma(\gamma) = W(\gamma)$. For any \SO\ decomposition $\Omega = (\gamma_i)_{i=1}^k$, we may conclude that $wn(x, \gamma) = \sum_{i=1}^k wn(x, \gamma_i)$ by \obsref{wnDecomp}.
    Thus, $wn(x, \gamma) = \sum_{i=1}^k wn(x, \gamma_i) \geq 0$ since
each $\gamma_i$ is positive consistent, as a positive \SO\ curve.

{\bf \ref{EIB:i} $\Leftrightarrow$ \ref{EIB:iv}:}
    If $\gamma$ has a well-behaved left sense-preserving nullhomotopy $H$
with exactly $k$ $\text{I}_a$-moves, then $H$ comes naturally with an associated \SO\ decomposition $\Omega$
of $\gamma$ with $|\Omega| = k$, and $\whit{\gamma} = k>0$ by \lemref{WHSSD}.
We now show that $\sigma(\gamma) = W(\gamma)$. Consider the reversal
$\overline{H}$ from the constant curve $\gamma_{p_0}(t) = p_0$ to
$\gamma$. Then, $\overline{H}$ is right sense-preserving and by \lemref{sense-preserving-monotone} the
function $a(i) = wn(x, \overline{H}(i, \cdot))$ is monotonically
increasing for any~$x \in \mathbb{R}^2$.  Since $wn(x, \gamma_{p_0}) =
0$ for all~$x \in \mathbb{R}^2$, we have that $wn(x, \gamma) \geq 0$
for all~$x \in \mathbb{R}^2$. Thus,~$\gamma$ is an interior boundary.
Conversely, if $\gamma$ is a positive interior boundary, then $\obs(\gamma)=0$ and by
\lemref{sensePreservingOptimal}, and since $\gamma$ is positive, $H$ is
left sense-preserving.
Again, by \lemref{WHSSD}, $\whit{\gamma} = j$,
where $j$ is the number of $\text{I}_a$-moves in any well-behaved nullhomotopy $H$ of
    $\gamma$. Hence, we must have $j = k$, as desired.
\qed
\end{proof}

\subsection{Equivalences of Self-Overlapping Curves}
In this section, we study different characterizations of \SO\ curves and
show their equivalence in \thmref{ESO}, which also shows that \SO\ curves are $1$-interior boundaries.

First we describe a geometric formulation of \SO ness,
inspired by the work of Blank and Marx \cite{BL67,Marx}.
%
%\begin{definition}
Let $\gamma \in \C$ be \SO. Let $P:[0,1]
    \rightarrow \mathbb{R}^2$ be a simple path
so that $P(0)=q=\gamma(t_q)$ and $P(1)=p=\gamma(t_p)$ lie on $[\gamma]$ but are not vertices of $\gamma$.
Without loss of generality, assume $t_q<t_p$.
Let $\widetilde{P} := \gamma|_{[t_q, t_p]}$, and suppose that
(1)~$P \cap \widetilde{P} = \{p, q\}$,
(2)~$C=\widetilde{P} * \overline{P}$ is a simple closed curve, and
(3)~$C$ is positively oriented; see \figref{BlankCut} as well as \figref{NiceSO}.
Then we call $P$ a \myemph{Blank cut} of $\gamma$. By cutting along $P$,
$\gamma$ is split into two curves of strictly smaller complexity, $\gamma_1$ and $C$.
%% We write $\gamma \stackrel{P}{\longrightarrow} \gamma_1$. We call the cut $P$ \myemph{effective}
%% when $\gamma_1$ is \SO. If we can find a sequence $(P_i)_{i=1}^k$ of effective
%% Blank cuts such that $\gamma \stackrel{P_1}{\longrightarrow} \gamma_1 \stackrel{P_2}{\longrightarrow} \cdots \stackrel{P_k}{\longrightarrow} \gamma_k $
%% and $\gamma_k$ is a simple positive \SO\ curve, then we say
%% say $\gamma$ admits the \myemph{Blank cut~decomposition} $(P_i)_{i=1}^k$.
We call a sequence $(P_i)_{i=1}^k$ of Blank cuts a  \myemph{Blank cut~decomposition} if the final curve is a simple
positively oriented curve.

\begin{figure}[ht]
\centering
\includegraphics[width = .85\textwidth]{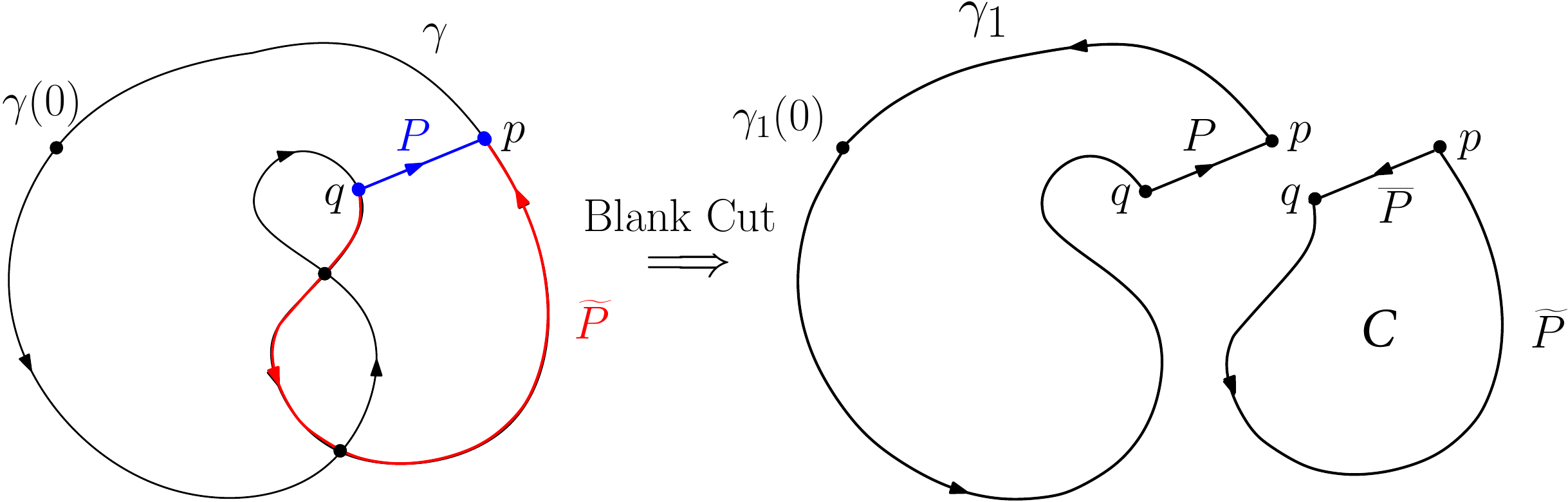}
\caption{A Blank cut on a small \SO\ curve.}
%The simple curve $\gamma_1$ produced by the cut is shown
%as well as the region $R$ removed from the interior by the cut.}
\label{fig:BlankCut}
\end{figure}

%We are now prepared to see the variety of equivalent formulations of SO curves.

\begin{theorem}[Equivalent Characterizations of Self-Overlapping Curves]
%[Equivalent Definitions of SO Curves,~\cite{BL67,fkw_mahncp_17,Marx,SVW92}]
\label{thm:ESO}
Let $\gamma \in \C$.
%have positive outer basepoint.
Then the following are equivalent:
\begin{enumerate}
\item (Analysis)
There is an immersion $F: D^2 \rightarrow \mathbb{R}^2$ so that $F|_{\partial D^2} = \gamma$.\label{ESO:i}
\item (Geometry)
$\gamma$ admits a Blank cut decomposition. \label{ESO:ii}
\item (Geometry/Topology) $\gamma$ is a $1$-interior boundary, i.e., \SO. 
\label{ESO:iii}
\item (Topology)
$\gamma$ admits a left-sense preserving nullhomotopy $H$ with exactly one $\text{I}_a$-move.\label{ESO:iv}
\item (Analysis) $\gamma$ is a Titus interior boundary with $\whit{\gamma} = 1$.
\label{ESO:v}
\end{enumerate}
\end{theorem}

\begin{proof}
By \propref{EIB:iii} in \thmref{EIB}, \SO\ curves are
$1$-interior boundaries, since any \SO\ curve $\gamma$ has the trivial \SO\
decomposition $\Omega = (\gamma)$. Thus, we have already established
$\ref{ESO:i} \Leftrightarrow \ref{ESO:iii} \Leftrightarrow \ref{ESO:iv} \Leftrightarrow\ref{ESO:v}$
in \thmref{EIB}. We now prove
$\ref{ESO:ii}\Leftrightarrow \ref{ESO:iv}$. Any Blank cut $P$ can be
performed by a left sense-preserving homotopy that deforms
$\widetilde{P}$ to $P$. Hence the Blank cut decomposition corresponds
to a left sense-preserving homotopy to a simple positively oriented
curve. Finally we perform a single $\text{I}_a$-move to complete a
left sense-preserving nullhomotopy of $\gamma$.
Conversely, let $\gamma$ have a left sense-preserving nullhomotopy
$H$.
From $\ref{ESO:iii} \Leftrightarrow \ref{ESO:iv}$
we know that every intermediary curve $\gamma_i = H(i, \cdot)$ 
is \SO\, since the subhomotopy $H_i = H|_{[i,1] \times [0,1]}$ is a left sense-preserving 
nullhomotopy of $\gamma_i$ with one $\text{I}_a$ move. As $H$ ends with a $\text{I}_a$ move, 
we may select a subhomotopy $H'$ such that
$\gamma \homotopyarrow{H'} C$, where $C$ is a simple \SO\ curve. Moreover,
we see that $H=H'+H''$, where the unique $\text{I}_a$-move
of $H$ occurs during $H''$. Thus, $H'$ is regular, i.e., consists of a sequence of homotopy
moves only of types $\text{II}_a$, $\text{II}_b$, or
$\text{III}$, which deform $\gamma$ to $C$.  Each of these
homotopy moves can be performed by a Blank cut, as shown
in \figref{BlankCut_Titus}. Since all of the intermediary curves are \SO,
this induces a Blank cut decomposition.
%we know each Blank cut is effective. 
\qed
\end{proof}

\begin{figure}[htbp]
    \centering
    \parbox{0.40\textwidth}{
        \centering
        \includegraphics[width=.40\textwidth]{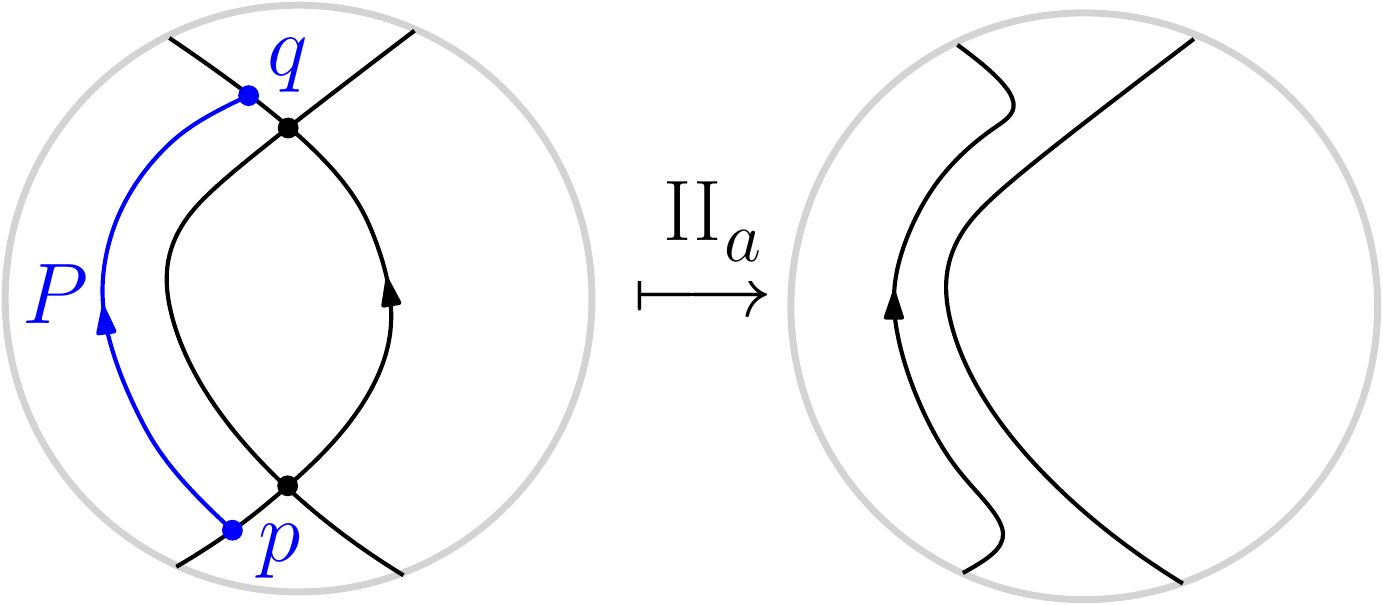}
    }
    \quad
    \begin{minipage}{0.40\textwidth}
        \centering
        \includegraphics[width=\textwidth]{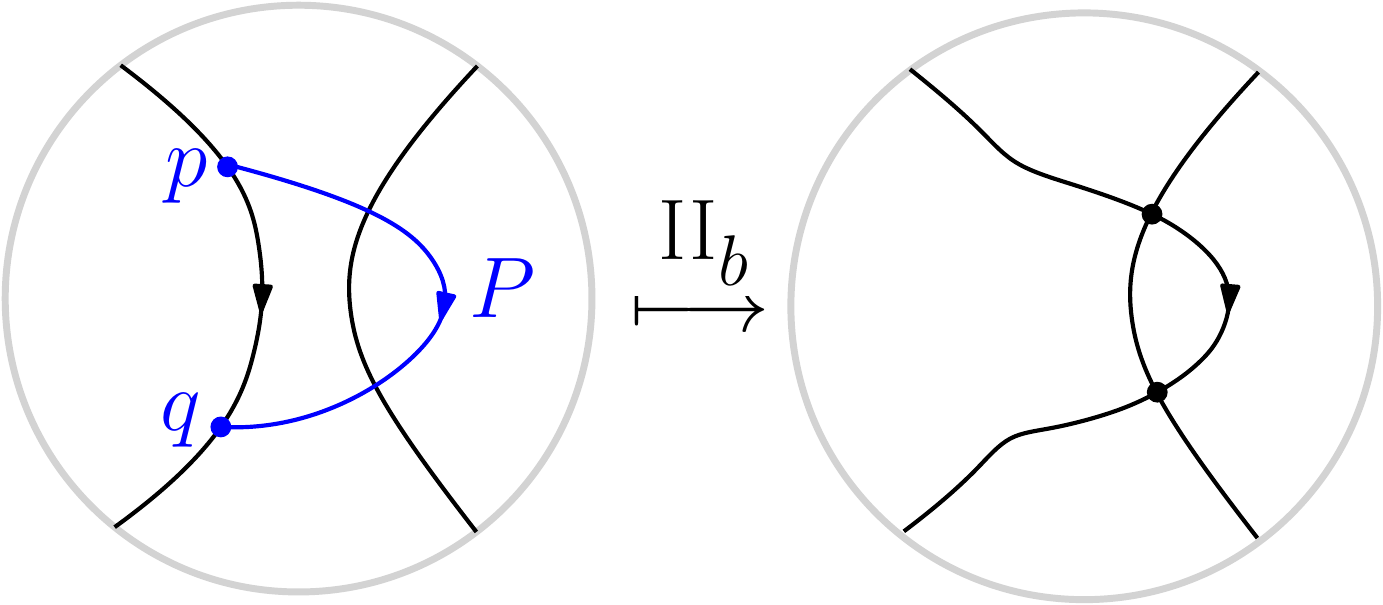}
    \end{minipage}\\[2ex]
    ~
    \begin{minipage}{0.40\textwidth}
        \centering
        \includegraphics[width=\textwidth]{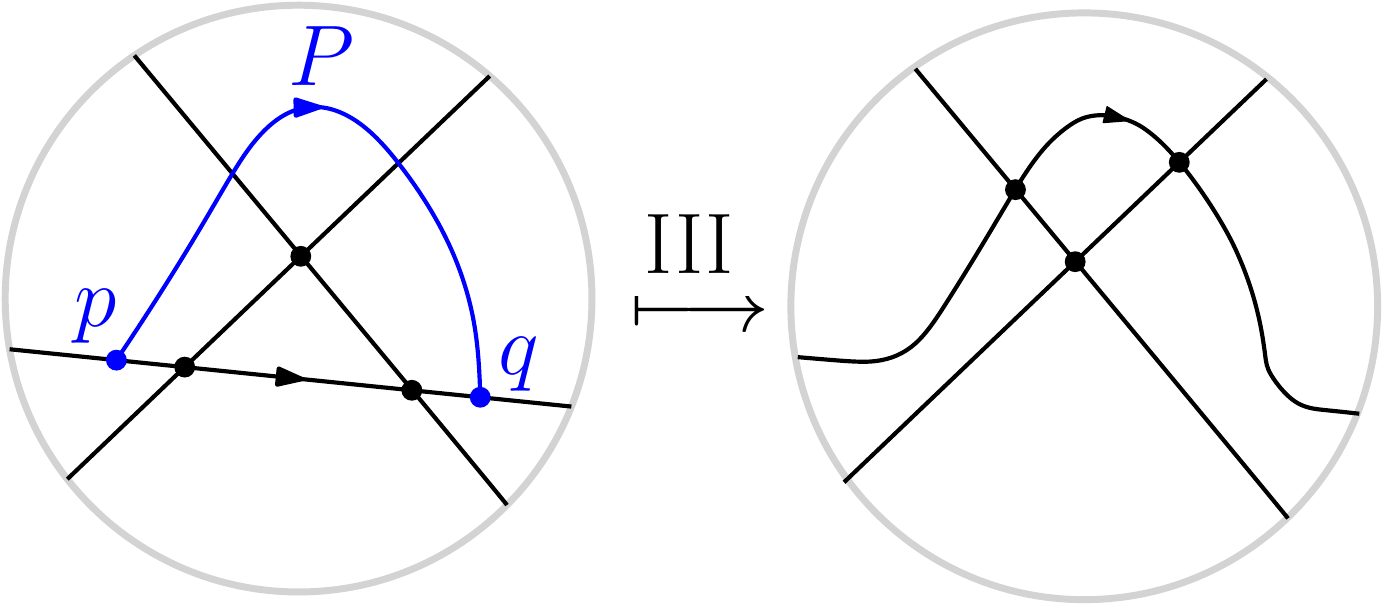}
    \end{minipage}
     \caption{Homotopy moves $\text{II}_a$, $\text{II}_b$, and $\text{III}$ each correspond to a Blank cut (shown in blue).}
%The curve remains unchanged outside the grey circle. Only one strand of the curve is labeled in each case as the orientation on the other strands is arbitrary.}
    \label{fig:BlankCut_Titus}
\end{figure}

The following two lemmas provide useful properties of \SO\ curves, the
first of which was proved in \cite[Theorem 5]{TT61}.
\begin{lemma}[Empty Positively Oriented Loop]\label{lem:EmptyLoop}
Let $\gamma \in \C$ have a positive outer basepoint and
    an empty positively oriented loop. Then,~$\gamma$ is not \SO.
\end{lemma}

We conclude this section with a simple yet powerful lemma. 
 
\begin{lemma}[Sense-Preserving Homotopies]\label{lem:RSP_homotopy}
Let $H$ be a regular homotopy with $\gamma \homotopyarrow{H} \gamma'$.
\begin{enumerate}
\item If $H$ is right sense-preserving and $\gamma$ is \SO, then $\gamma'$ is \SO.\label{RSP:i}
\item If $H$ is left sense-preserving and $\gamma$ is not \SO, then $\gamma'$ is not \SO.\label{RSP:ii}
\end{enumerate}
\end{lemma}

\begin{proof}
To prove \ref{RSP:i}, assume $\gamma$ is \SO. Then it has a left sense-preserving nullhomotopy
  $H'$ by \thmref{ESO}. Let us reverse our given
  homotopy $H$ to obtain $\overline{H}$ by $\overline{H}(s,t) =
  H(1-s,t)$. Then we note that the concatenation $H'' = \overline{H} +
  H'$ is a left sense-preserving nullhomotopy for~$\gamma'$. Since
  sense-preserving homotopies are optimal, $\sigma(\gamma') =
  W(\gamma')$. Also, as $H''$ is regular, $W(\gamma')=W(\gamma)=1$.
Applying \thmref{ESO} again, we conclude that $\gamma'$ is \SO.
Part \ref{RSP:ii} follows by contrapositive with a single application of \ref{RSP:i}: if $\gamma'$ were \SO\ then $\gamma$
must be \SO\ as well.
\qed
\end{proof}

\subsection{Zero Obstinance Curves}\label{sec:idealCurves}

In this section, we classify curves $\gamma\in\C$ with \myemph{zero obstinance},
$\obs(\gamma):=\sigma(\gamma)-W(\gamma)=0$. See Figures \ref{fig:smallerCurveLattice}
and \ref{fig:Ideal} for examples of zero-obstinance curves.
We show that just as interior boundaries
can be decomposed into \SO\ curves, so too can zero-obstinance curves be decomposed into
interior boundaries.

\begin{figure}[htbp]
\centering
\includegraphics[width=0.60\textwidth]{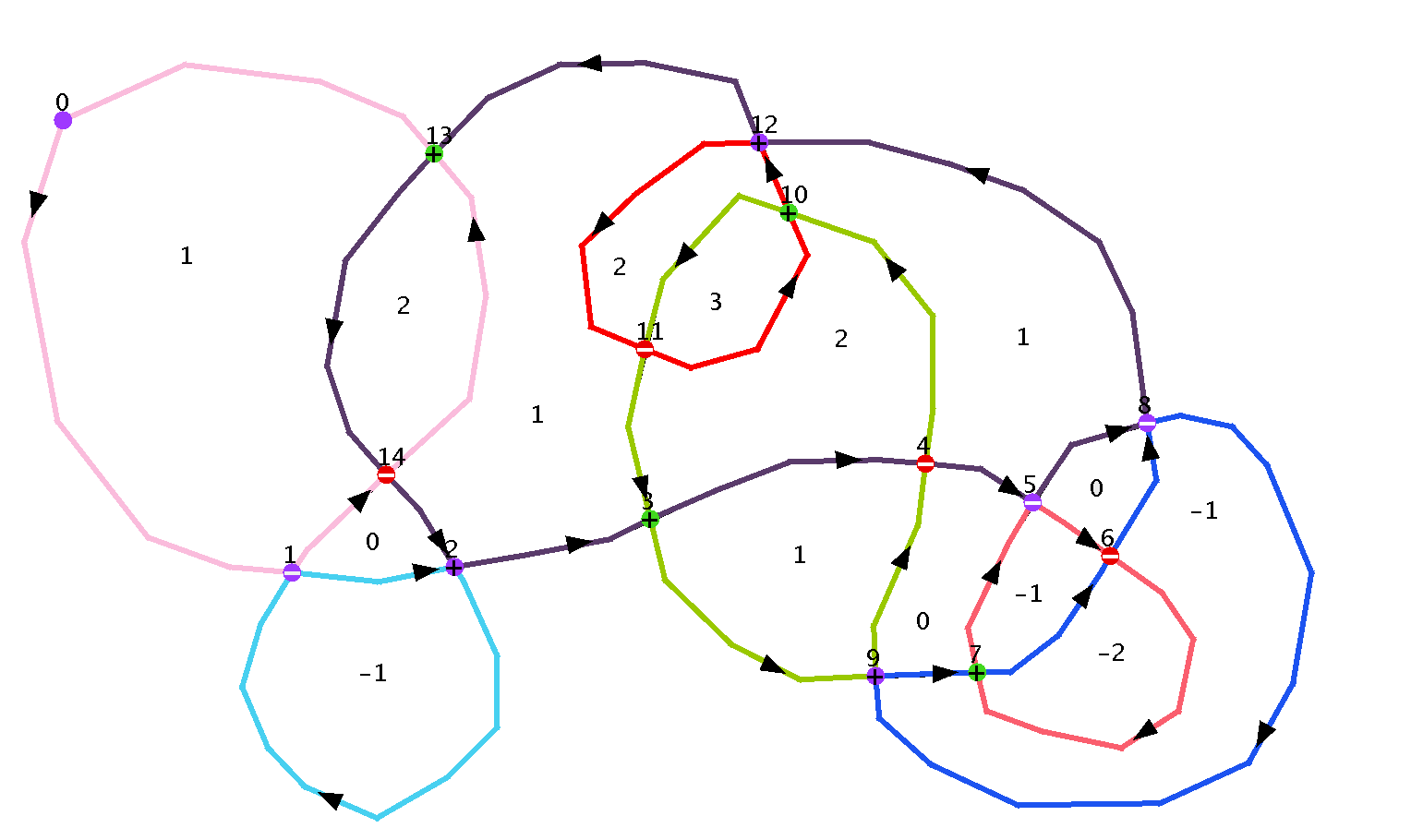}
\caption{A zero obstinance curve, with its minimum homotopy decomposition, and winding numbers shown. Each curve in the decomposition is \SO\ and shown in a different color. The vertices with labels $1,2,5,8,9$ are sign-changing. %Figure created with \cite{Program}.
%can all be seen to be anchor points.
}
\label{fig:Ideal}
\end{figure}

If a curve $\gamma$ has zero obstinance, then there is a nullhomotopy
$H$ which sweeps each face $F$ on $G(\gamma)$ exactly $wn(F, \gamma)$
times. Note that such a homotopy $H$ is necessarily minimal by Lemma \ref{lem:minHomotopy_and_windingAreas}. Intuitively, this implies that the homotopy
$H$ should be locally sense-preserving. We expect it to sweep leftwards on
positive consistent regions and rightwards on negative consistent
regions. Hence, we might expect regions of the curve where the winding numbers change from positive to negative to be especially
problematic.
Indeed, let $v \in V(\gamma)$ be incident to the faces $\{F_1, F_2, F_3,
F_4\}$. We call $v$ \myemph{sign-changing} when, as a multiset,
$\{wn(\gamma, F_1), wn(\gamma, F_2), wn(\gamma, F_3), wn(\gamma,
F_4) \} = \{-1, 0, 0, 1\}$; see Figures \ref{fig:Ideal} and \ref{fig:SignChanging}.
%\figref{SignChanging}.
%
%
\begin{figure}[ht]
	\centering
	\includegraphics[width=0.22\textwidth]{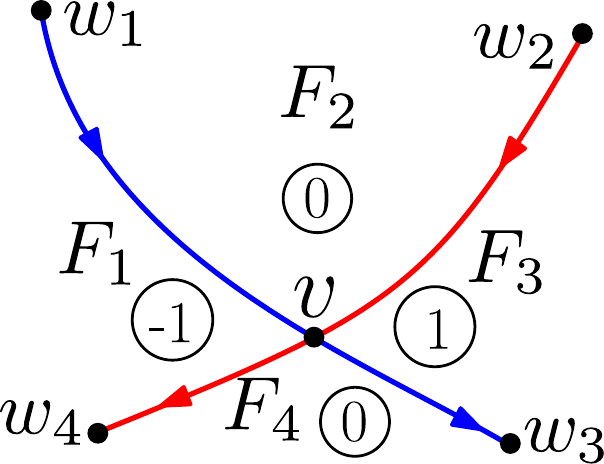}
	\caption{A sign-changing vertex $v$. The winding numbers of the faces incident to $v$, are up to cyclic reordering, -1, 0, 1, 0.}
	\label{fig:SignChanging}
\end{figure}

\begin{theorem}[Zero Obstinance Characterization]\label{thm:Ideal}
Let $\gamma \in \C$ and let $\mathscr{S}$ be the sign-changing
    vertices of $\gamma$. Then $\obs(\gamma)=0$ iff no two vertices in
    $\mathscr{S}$ are linked and any direct split subcurve decomposition
    $\Omega$ with vertex set $V(\Omega) = \mathscr{S} \cup \{p_0\}$
    contains only interior~boundaries.
\end{theorem}

\begin{proof}
Suppose $\obs(\gamma) = 0$. 
By definition, any zero obstinance curve with consistent winding numbers
must be an interior boundary and $\mathscr{S}=\emptyset$. Hence,
suppose $\gamma$ is inconsistent so that $\mathscr{S} \neq \emptyset$. We claim any sign-changing
vertex $v$ is anchor point of every well-behaved minimum homotopy $H$ of $\gamma$ of the form guaranteed
by \thmref{mhd}. 
Let us now proceed by contradiction. Suppose $v \in V(\gamma)$ is a sign-changing vertex with incident 
faces labeled as in \figref{SignChanging} such that $v$ is not an anchor point of a minimum
homotopy $H$ for $\gamma$. Write $\Gamma(H)$ as the \SO\ decomposition of $H$. 
As $\gamma$ has $\obs(\gamma) = 0$, we know that $W(\gamma) = \sigma(\gamma) = A(H)$. In particular,
the homotopy $H$ sweeps each face $F \in G(\gamma)$ precisely $wn(\gamma, F)$ times. Of course,
since our homotopy $H$ consists of a sequence of nullhomotopies of \SO\ subcurves, this means each face 
$F$ must lie in the interior of $wn(\gamma, F)$ distinct \SO\ subcurves $C \in \Gamma(H)$. 
In particular, if either face $F_2, F_4$ incident to $v$ is contained in the interior of any curve $C \in \Gamma(H)$,
we have a contradiction. Let us now examine the edge $e = (v, w_3)$. This edge
must lay on precisely one subcurve $C \in \Gamma(H)$ by our definition of a direct split subcurve decomposition. 
We now have two cases.

\textbf{Case 1}: $C$ is positively \SO. We now recall that for a positive \SO\ curve, the interior of the curve always 
lies, locally at each edge, to the left. Since $v$ is not an anchor point of $H$, it must be the case that $C$ also contains the
edge $e_1 = (w_1, e)$. As the face $F_2$ lies to the left of $e_1$, this implies $F_2 \subset \intt{\; C}$. 

\textbf{Case 2}: $C$ is negative \SO. Here, we use that the interior of a negative \SO\ curve lies locally to the right. 
In this case, we see that $F_4$, lying to the right of edge $e_1$, satisfies $F_4 \subset \intt{ \; C}$. 

We conclude that all 
sign-changing vertices are anchor points of $H$. This is only possible
if none of the sign-changing vertices link each other. Now, let $\Theta = (\gamma_i)_{i=1}^k $ be any direct split subcurve
decomposition with $V(\Theta) = \mathscr{S} \cup \{p_0\}$. We claim that each curve $\gamma_i \in \Theta$
is an interior boundary. It suffices to prove this claim for any decomposition $\Omega \sim \Theta$. Hence, we may 
select the unique representative $\Omega$ from the equivalence class of $\Theta$ such that the ordering of $\Omega$ is compatible
with the ordering of $\Gamma(\gamma)$. 
Since each sign-changing vertex is an anchor point of $H$, it follows that $\Gamma(H)$ is a refinement of the decomposition $\Omega$.
Thus, by \thmref{mhd}, there is a subhomotopy $H_i$ of $H$ which is a nullhomotopy of $\gamma_i$. As subhomotopies
of a minimum homotopy $H$, each $H_i$ must be minimum as well, $A(H_i) = \sigma(\gamma_i)$. Observe that
$$ \sum_{i=1}^k W(\gamma_i) = W(\gamma) = A(H) = \sum_{i=1}^k A(H_i) = \sum_{i=1}^k \sigma(\gamma_i).$$
So, we must have equality, $\sigma(\gamma_i) = W(\gamma_i)$, for every curve in the decomposition. By Lemma \ref{lem:minHomotopy_and_windingAreas},
this means each $\gamma_i \in \Omega$ has $\obs(\gamma_i) = 0$. Hence, if each $\gamma_i$ is consistent, they are all interior 
boundaries, by definition. Of course, if some $\gamma_i$ were inconsistent, then the winding numbers would change somewhere along the curve. Wherever 
the winding numbers of $\gamma_i$ change, we will see a sign-changing vertex $u \in V(\gamma_i)$. But since $u$ is not the basepoint of $\gamma_i$, this is a contradiction. Indeed,
by the fact that $\mathscr{S} \cup \{p_0\} = V(\Omega)$, 
no sign-changing vertex can be a crossing point on a curve $\gamma_i \in \Omega$. 

Conversely, suppose no sign-changing vertices link each other and that each decomposition $\Omega = (\gamma_i)_{i=1}^k$ of $\gamma$ with vertex set $V(\Omega) = \mathscr{S} \cup \{p_0\}$ contains only interior boundaries. Then let $H_{\Omega}$ be the homotopy associated to $\Omega$. Thus, we have
$W(\gamma) = \sum_{i=1}^k W(\gamma_i) = \sum_{i=1}^k \sigma(\gamma_i) = A(H)$
We conclude $H$ is optimal and $\sigma(\gamma) = W(\gamma)$. Thus, $\obs(\gamma)=0$.
\qed
\end{proof}

\section{Wraps and Irreducability}\label{sec:wraps}
In this section, we show (Theorems \ref{thm:WRB} and \ref{thm:WRB2}) that wrapping around a curve $\gamma$ until its
obstinance is reduced to zero results in an interior boundary. This
key result is used to prove sufficient
combinatorial conditions for a curve to be \SO\ based on the Whitney
index of the curve and its direct splits (Theorems \ref{thm:ISO} and \ref{thm:THSO}).

\subsection{Wraps}
Let us now define the construction of the wrap of a curve.
%Note that
%the result of this process is unique only up to signed intersection
%sequence. On the other hand, due to work of Titus (Theorem 3
%in \cite{TT61}) we know that being an interior boundary is invariant
%with respect to signed intersection sequence, assuming positive outer
%basepoint.
%\begin{definition}
Let $\gamma \in \C$, and let $I$ be its
signed intersection sequence. Form $I'$ by incrementing each
label by one and removing the occurrences of $0$ corresponding to the basepoint.
If $\gamma$ has a positive outer basepoint $\gamma(0)$, then its \myemph{(positive) wrap} $\mathemph{Wr_+(\gamma)}$ is the unique (class) of curves with signed
intersection sequence  $0,1_+, I', 1_-, 0$.
This corresponds to gluing a simple positively oriented curve $\alpha$
to $\gamma$ at $\gamma(0)$, where
the interior $\text{int}(\alpha)\supseteq [\gamma]$;
%onto $\gamma$ at the old basepoint $\gamma(0)$, along with placing
the new basepoint $p_0 = Wr_+(\gamma)(0)$ is on $\alpha$. See \figref{wrap}. The \myemph{negative wrap}
$\mathemph{Wr_-(\gamma)}$ is
defined analogously if $\gamma$ has a negative outer basepoint.
We write
$Wr_+^k(\gamma)$ for the curve achieved from $\gamma$ by wrapping $k$ times.

\begin{figure}[htbp]
    \centering
    \includegraphics[height = .32\textwidth]{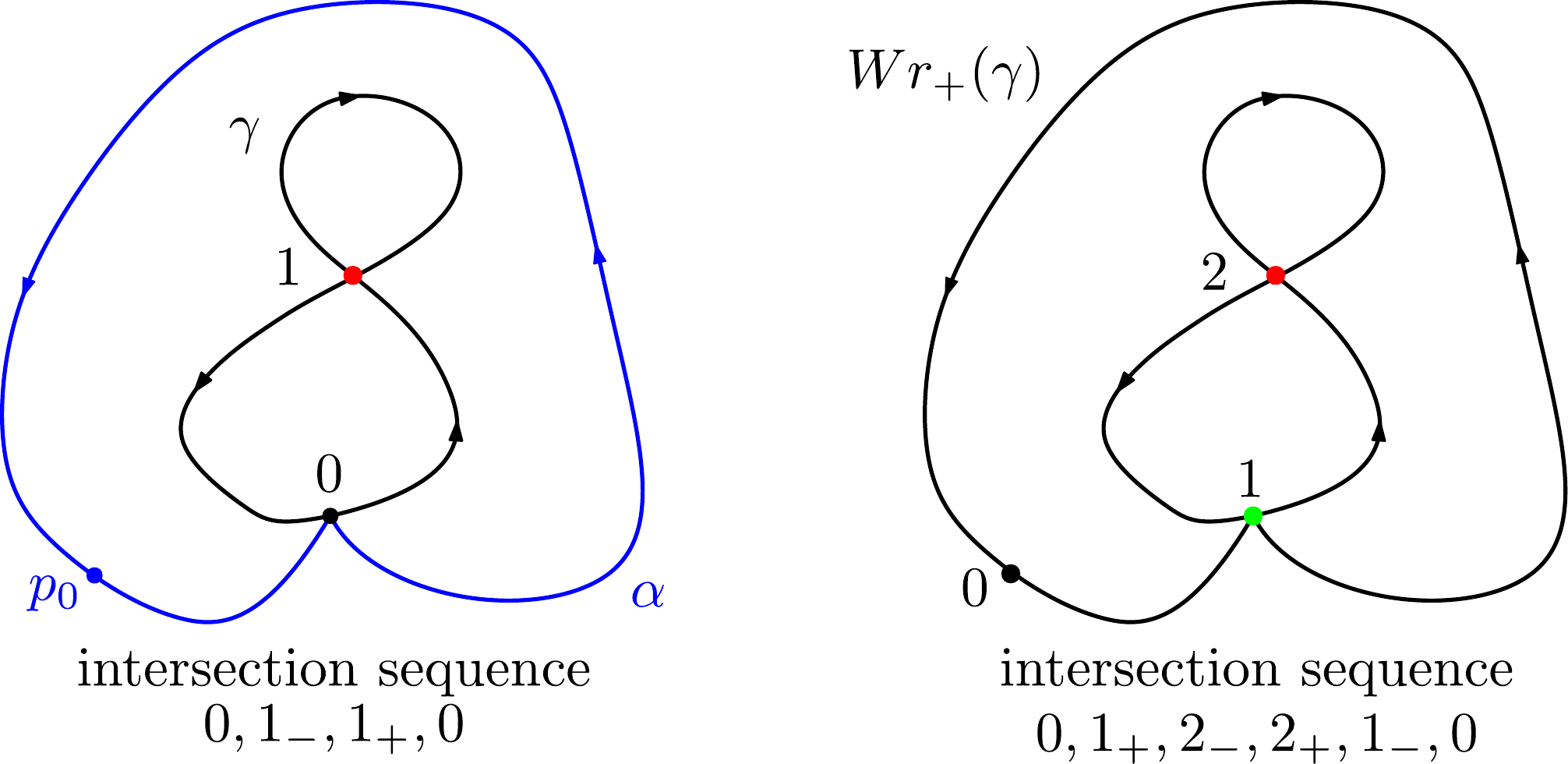} % Reference curve for scaling
    \caption{Left: A curve $\gamma$ (shown in black) with positive outer basepoint, and the curve $\alpha$ (shown in blue).
    Right: The \emph{positive} wrap $Wr_+(\gamma)$. While $\gamma$ is not \SO, its wrap $\Wr_+(\gamma)$ is \SO.}
    %The (unsigned) intersection sequences of both curves are shown. Note that $\gamma(0)$ becomes a positive vertex after smoothing.}
    \label{fig:wrap}
\end{figure}

To wrap a curve in the direction opposed to the sign of the basepoint, we must
be more careful. Without loss of generality, we describe the construction of $Wr_{-}(\gamma)$ when
$\gamma$ has a positive outer basepoint.
%$\sgn(\gamma(0))= +1$.
Perform a $\text{I}_b$-move to add a simple loop $\tilde{\gamma}$ of the opposite orientation tangent to the basepoint $\gamma(0)$. Let $\gamma'$
be the curve after the $\text{I}_b$-move, with a basepoint  chosen to lie on $\tilde{\gamma}$.
We then define $Wr_-(\gamma) = Wr_+(\gamma')$. See \figref{wrapOppo}.
%\end{definition}

%\textbf{Remark}: The wrap does not define a new plane curve, but rather defines a curve up to signed intersection sequence. See \figref{wrap} and \figref{wrapOppo}.

\begin{figure}[htbp]
    \centering
    \includegraphics[height = .388\textwidth]{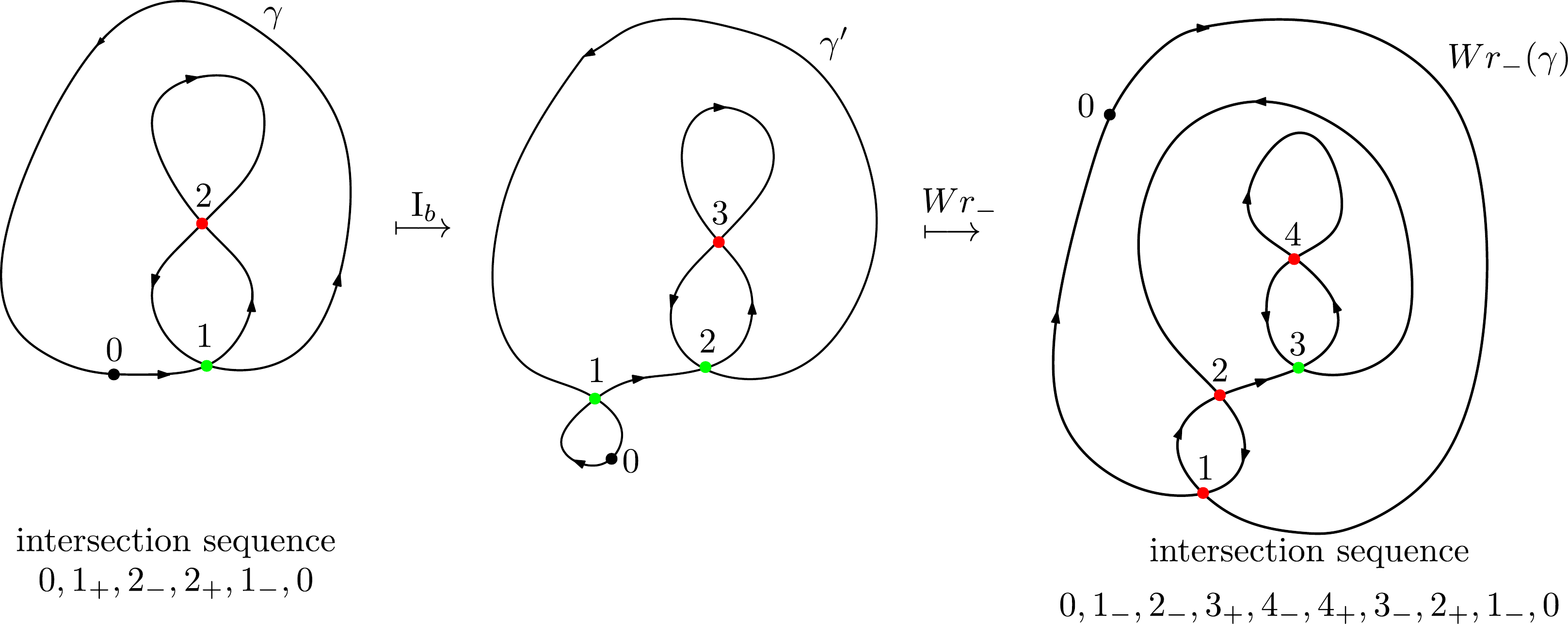}
    \caption{A curve $\gamma$ with positive outer basepoint and
its transformation into its \emph{negative} wrap $Wr_{-}(\gamma)$. First, we perform a $\text{I}_b$-move and then
wrap normally on $\gamma'$.}
    \label{fig:wrapOppo}
\end{figure}

Clearly one can always wrap any curve $\gamma\in\C$
a sufficient number of times to make $Wr_+^k(\gamma)$ positive
consistent. Indeed, setting $k$ to be the maximum depth across all
faces in $G(\gamma)$ suffices.
%
%\footnote{See \thmref{CurveClasses} and \figref{CurveClasses} in \appendixref{lattice}
%for more details on the relationship between different curve classes, including
%interior boundaries and consistent curves.}
%
On the other hand, it is not at all obvious that wrapping
always turns a curve into an interior boundary.
We prove in \thmref{WRB}  that,
in fact, positively wrapping always transforms a curve $\gamma\in\C$
with positive outer basepoint into a positive interior boundary. Thus
one can think of wrapping as a rectifying operation with respect to
minimum homotopy, as it always eventually removes all obstinance.

\subsection{Simple Path Decompositions}\label{sec:SPD}

We now describe another type of decomposition for $\gamma \in \C$ that
we will need for proving \thmref{WRB2}.
First we prove a simple lemma which states
that a curve $\gamma\in\C$ with an outer basepoint has an outwards loop.

\begin{lemma}[Existence of an Outwards Loop]\label{lem:OL}
Let $\gamma \in \C$ have an outer basepoint. Then if $\gamma$ is non-simple, it has an outwards loop.
\end{lemma}

\begin{proof}
Let $v$ be the first self-intersection of $\gamma$. Then $\gamma_v$ is a loop.
Write $\gamma^{-1}(v) = \{ t, t^*\}$, where $t < t_*$. 
Since $ \gamma(0)$ lies outside of $\intt(\gamma_v)$, as an outer basepoint,
we note that if $\gamma_v$ were inwards, the path
$P = \gamma_{[0,t_1]}$ would cross $[\gamma_v]$ to get from outside
the simple curve
to inside it.  This is then a contradiction, for if the crossing
occurred at a point $q$ on $[\gamma_v]$, then $q$ would be the first
self-intersection of $\gamma$. Indeed, we would reach $q$ a second time
before we reach $v$ a second time. Thus, $\gamma_v$ is outwards.
\qed
\end{proof}

Now, let $t_1<t_1^* \in [0,1]$ be the smallest value so that
$\gamma_1=\gamma_{[t_1, t_1^*]}$ is a loop. We call this
the \myemph{first loop} of $\gamma$.
Since $t_1^*$ is the first time
that $\gamma$ self-intersects, we know that such a loop
is outwards by the argument in \lemref{OL}.
The sub-intersection sequence from $r= \gamma(0)$
until the second occurrence of $p_1=\gamma(t_1)= \gamma(t_1^*)$ is of the following form:
$r \cdots p_1 \cdots p_1$, where no intersection point occurs twice
before the second occurrence of $p_1$.
Let $P_1$ be the path from $r$ until the first occurrence of $p_1$, i.e., 
$\gamma_{[0, t_1^*]} = P_1 * \gamma_1$.
We continue by recursively defining 
$$t_i^* = \sup \{t \in [0,1] \; | \;
t>t_{i-1}^* \text{ and } \gamma_{[t_{i-1}^*, t]} \text{ is injective}\},$$
for $i>1$.
%; we stop if $t_i=1$.
Then $\gamma(t_i^*) = p_i$ must be a
basepoint of a loop and there exists a $t_i<t_i^*$ such that
$p_i=\gamma(t_i)=\gamma(t_i^*)$. We define $P_i=\gamma_{[t_{i-1}^*,t_i]}$
and $\gamma_i=\gamma_{[t_i,t^*_i]}$.
%to be the simple path from $p_{i-1}$ to $p_i$.
Let $k$ be the largest value such that $t^*_k<1$, 
and set $P_{k+1}=\gamma_{[t_k^*,1]}$.
Then $\gamma = P_1 * \gamma_1 * \cdots * P_k
*\gamma_k * P_{k+1}$, where $*$ denotes concatenation.
The \myemph{simple path decomposition} of $\gamma$ is the
sequence $(P_i,\gamma_i)_{i=1}^k$; see \figref{SPD}.

\begin{figure}[ht]
\center
\includegraphics[height=0.31\textwidth]{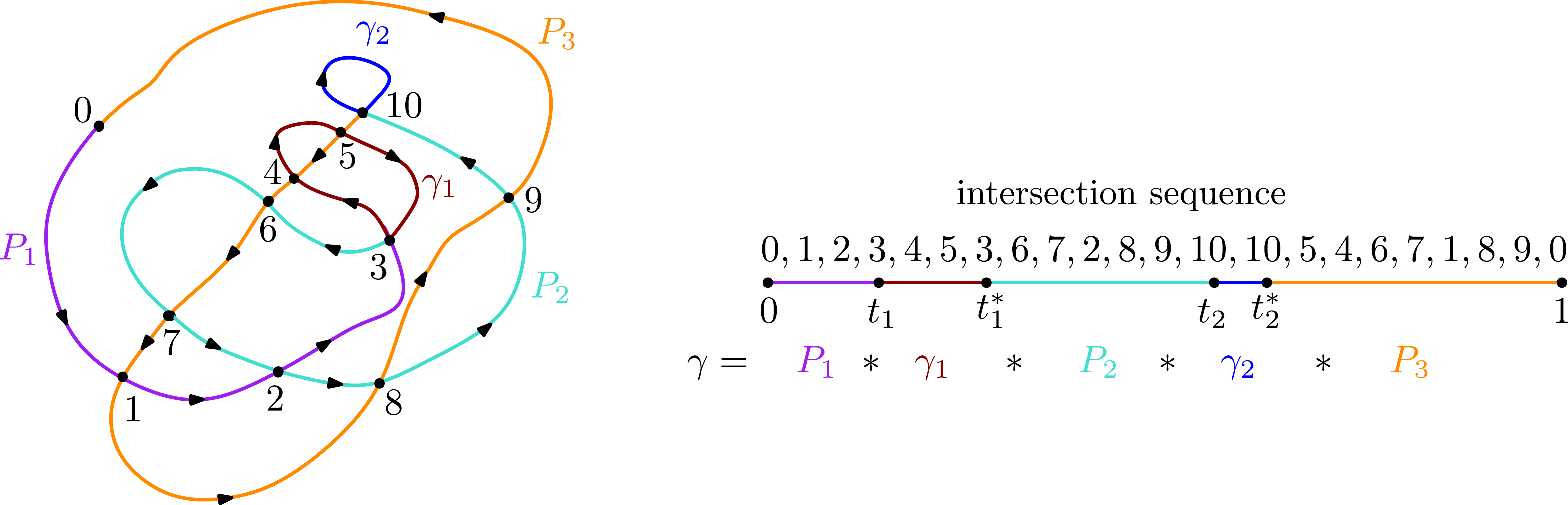}
\caption{A strongly irreducible curve with its simple path decomposition shown.
}
%The intersection sequence is shown on the side, including the simple path decomposition.}
\label{fig:SPD}
\end{figure}

\subsection{Wrapping Resolves Obstinance}
\label{sec:wrappingResolvesObstinance}
We are now equipped to prove our second main result on wraps.
%; it will be generalized in \thmref{WRB2}.

%the main theorem of this section; it will be generalized in \thmref{WRB2}.
%We are now equipped to prove our second main result on wraps. 

\begin{theorem}[Wrapping Resolves Obstinance]\label{thm:WRB}
Let $\gamma \in \C$ have positive outer basepoint. Then there is a positive integer $k$
so that $\obs(\Wr_+^k(\gamma)) = 0$. Moreover, $\Wr_+^k(\gamma)$ is a positive interior boundary.
\end{theorem}

\begin{proof}
Let $l$ be the number of negative vertices in $V(\gamma)$. Set $k = l+1$.
We claim that $\Wr_+^k(\gamma)$ is an interior boundary. We will show this by iteratively
constructing a left sense-preserving nullhomotopy $H$ for $\gamma$. By \propref{EIB:iv} of \thmref{EIB} it then follows that $\gamma$ is a positive interior boundary and $\obs(\gamma)=0$.

\begin{figure}[thb]
\center
\includegraphics[width=0.4\textwidth]{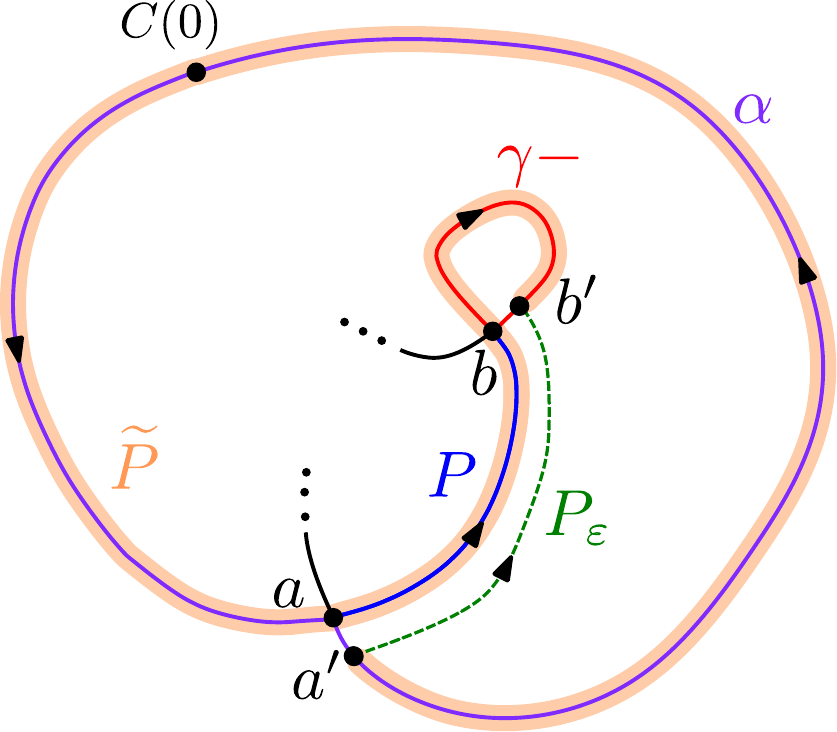}
\caption{The combinatorial structure necessary to apply balanced loop deletion: a wrapped curve, with outer wrap $\alpha$
and a negatively oriented loop $\gamma_-$ as first loop in the simple path decomposition. }
\label{fig:BLD}
\end{figure}

We first introduce a trick that we call \myemph{balanced loop deletion}.
See \figref{BLD}, where all of the following objects are shown.
Suppose that $C \in \C$ is a curve that is positively wrapped, $C = Wr_+(C')$ for some curve $C' \in \C$,
and suppose  that the first loop $\gamma_-$ (shown in red) in the simple path decomposition of $C$ is negatively oriented.
Let $b=C(t_b)=C(t_b^*)$, with $t_b < t_b^*$, be the basepoint of $\gamma_-$.
Balanced loop deletion performs a left sense-preserving homotopy $H$
 so that $C \homotopyarrow{H} C \setminus \left(\alpha \cup \gamma_- \right)$,
where $\alpha$ (shown in purple) is the positive outer wrap on $C$.
%

%
 %, in the context of balanced loop deletion.
%Let $P$ be the simple path along the orientation of $C$ from the positive outer point $C(0)$ to $b$.
Let $P$ (shown in blue) be the simple subpath of~$C$  from $a=C(t_a)=C(t_a^*)$ to $b$, where $a$ is
the unique outer intersection point on $[C]$, i.e., the basepoint of the wrap $\alpha$, and  $t_a < t_a^*$.
For $\varepsilon>0$ sufficiently small, let $a' = C(t_a^* + \varepsilon)$ and $b' = C(t_b^* - \varepsilon)$  and let
$P_\varepsilon$ (shown in dashed green) be a simple path between $a'$ and $b'$ that is
$\varepsilon$-close to $P$ in the Hausdorff distance.
%
%%  Let us take a right $\epsilon$ copy $P''$
%% of $P'$. Then, for $\epsilon$ sufficiently small, we can make an arbitrarily small deformation to $P''$
%% to transform it into a path $P'_{\epsilon}$ between points $a'$ and $b'$ which are on $C$ and
%% are in a right $\epsilon$ neighborhood
%% of $a$ and $b$, respectively. Let us take $t_a < t_a^*$ and $t_b < t_b^*$ so that $C(t_a) = C(t_a^*) = a$
%% and $C(t_b) = C(t_b^*) = b$. Then, since the direct split of $a$ is the negative loop $\gamma_-$ and the wrap
%% split of $b$ is the positive simple subcurve $\alpha$, we will have $a' = C(t_a^* + \varepsilon_1)$
%% and $b' = C(t_b^* - \varepsilon_2)$ for some small $\varepsilon_1, \varepsilon_2$.
%
Let $\widetilde{P}=C|_{[t_a^* + \varepsilon,t_b^* - \varepsilon]}$ (shown in thick beige) be
the simple subpath of~$C$ from $a'$ to $b'$.
Then $\widetilde{P}$ is the concatenation of
(i) the path from $a'$ to $a$ along $\alpha$,
(ii) the path~$P$ from $a$ to $b$, and
(iii) the path from $b$ to $b'$ along $\gamma_-$.
The path $\widetilde{P}$ is simple because each of these subpaths are simple
and none of them intersect each other since $b$ is the first
self-intersection point of the curve.
We now make a crucial observation: $\widetilde{P}*P_{\varepsilon}$
is a simple, positively oriented, closed curve. It follows that
we can perform a Blank cut along $P_\varepsilon$ that replaces
$\widetilde{P}$ on $C$ with the path $P_{\epsilon}$.
The effect of this cut on $C$ is that both the outer wrap~$\alpha$ and the
negatively oriented loop $\gamma_-$ are deleted, and the path $P$ is replaced by
$P_{\varepsilon}$.
%
%% For $\varepsilon$ sufficiently
%% small, the trade of $P'$ for $P'_{\varepsilon}$ will not affect the
%% portion of the signed intersection sequence to which $P'$
%% corresponds.
%% %
%% Thus, on the level of the intersection sequence, the
%% Blank cut $\widetilde{P}$ deletes both the direct split $C_b= \gamma_-$
%% and the wrap $C_{a^*} = \alpha$, and leaves everything else
%% unaffected.
%
This Blank cut can be performed by a left sense-preserving
homotopy, so we have established the existence of left
sense-preserving balanced loop deletion.

Now we construct a left sense-preserving nullhomotopy $H$ of
$\Wr_+^k(\gamma)$ by iteratively concatenating several left
sense-preserving subhomotopies, so $H= \sum_i H_i$. We proceed
inductively as follows. Suppose $H_1, \dots, H_{i-1}$ have been
defined and $\gamma_i$ is the current curve. Consider the first loop
$C_i$ in the simple path decomposition of $\gamma_i$. If $C_i$ is
positively oriented we let $H_i$ be the left sense-preserving nullhomotopy
that contracts this loop. Otherwise $C_i$ is negatively oriented
and we let $H_i$ be the homotopy performing balanced loop deletion.

We claim that we always have a wrap available to perform this balanced
loop deletion. Each homotopy $H_j$ for
$j=1\ldots i-1$ deletes one or two free subcurves of~$\gamma_{j}$. Therefore the signs of the remaining intersection points
are not affected.
%
%Indeed, all the vertices created by wrapping remain positive, and since
%the basepoint still lies on a wrap, the
%signs of the vertices originally on $\gamma$ are preserved as
%well.
%
Observe that if a vertex $v$ is the basepoint of an outwards loop
$\gamma_v$, then
%the sign $\sgn(v)$ is naturally given by the orientation of the loop:
$\sgn(v) = 1$ iff $\gamma_v$ is positively
oriented and $\sgn(v) = -1$ iff $\gamma_v$ is negatively
oriented; see \figref{LoopsSigns}.
\begin{figure}[ht]
\centering
\includegraphics[height  = .2\textwidth]{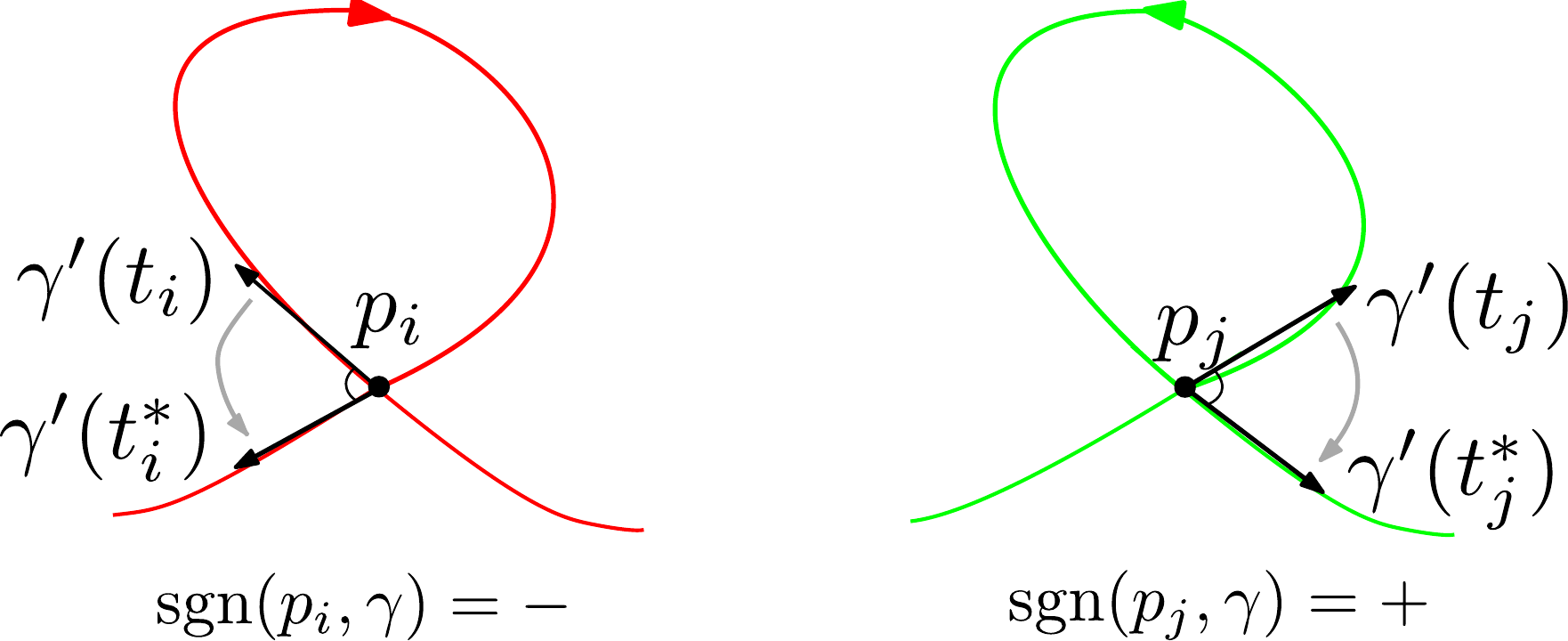}
\caption{Two outwards loops; negatively oriented (left) and positively oriented (right).}
%
%The basepoint of an outwards negatively oriented loop is negatively signed,
%while the basepoint of an outwards positively oriented loop is positively signed.}
\label{fig:LoopsSigns}
\end{figure}
By definition of $l$, we have $n_i \leq l$, where $n_i$ is
the number of negative vertices on $\gamma_i$.
Therefore there can be at most $l$ distinct integers $i_1,\ldots,i_l$
such that the first loop on $\gamma_{i_\nu}$  is negatively oriented,
by \lemref{OL} the first loop is always outwards.
Since $k=l+1$, we always have a wrap available on~$Wr_+^k(\gamma)$.

The process of constructing homotopies $H_i$ never gets stuck, and
$\abs{\gamma_{i+1}} < \abs{\gamma_i}$.  Therefore we must eventually
reach a point when the current curve $\gamma_{m}$ has $\abs{\gamma_m}
= 0$. We now show that this final curve $\gamma_m$ is positively
oriented.
Note that if $\gamma$ is simple then $\Wr_+^k(\gamma)$ trivially is a positive
$k$-interior boundary. So, assume $\gamma$ is not simple.

By the definition of wraps, the intersection sequence of $\Wr_+^k(\gamma)$ has the form
$0,1+,2+,\ldots,k+,I',k-,\ldots,1-,0$, where $I'$ is obtained from the signed intersection sequence of $\gamma$ by incrementing each label by $k$ and removing occurrences of $0$.
The basepoint of the first loop on $\gamma_1 = \Wr_+^k(\gamma)$ must
be a vertex from $\gamma$. Then $H_1$ modifies the intersection sequence by
removing two labels from $I'$ corresponding to this loop, and if the loop is negatively oriented then balanced loop deletion also removes
a pair $a+\ldots a-$ for the wrap. The same modification happens for each homotopy $H_i$ until $I'$ is empty, and the homotopies after that contract wraps $a+\ldots a-$ which are all positively oriented loops.
We know that $\gamma$ has at most $l=k-1$ negative vertices, hence
there can only be $k-1$ balanced loop deletions, but there are $k$
wraps. Thus, $\gamma_m$ must be a
%smoothed version of a
wrap
that $Wr_+^k(\gamma)$ added
to $\gamma$, so $\gamma_m$ is a positively oriented loop
which can be contracted to its basepoint using a final left sense-preserving homotopy
$H_m$.
This shows that $H = \sum_{i=1}^m H_i$ is a left sense-preserving nullhomotopy of $\Wr_+^k(\gamma)$ as desired.
\qed
\end{proof}

\begin{figure}[htbp]
    \centering
    \includegraphics[height = .32\textwidth]{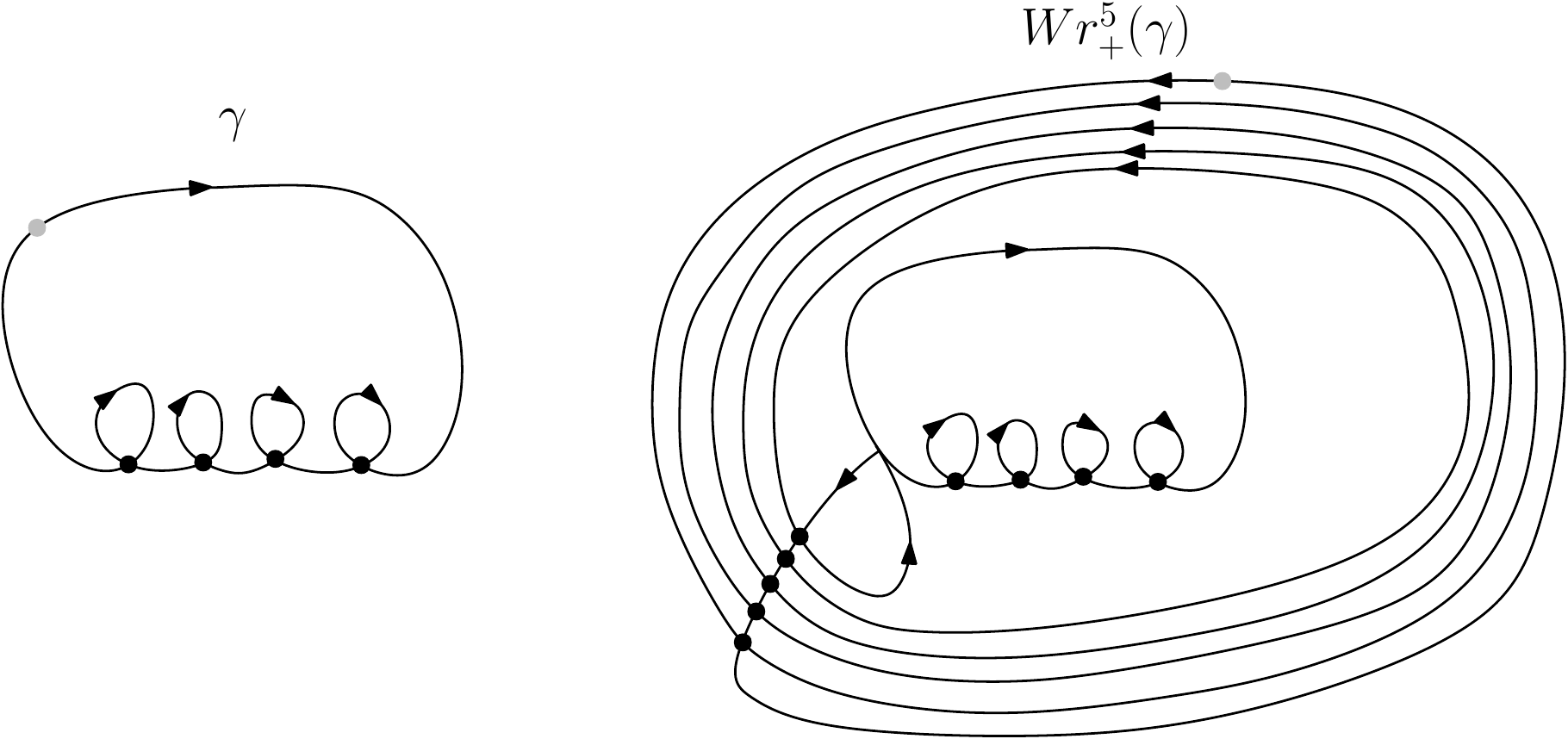} % Reference curve for scaling
    \caption{A curve $\gamma$ that requires $k=l=5$ wraps to resolve obstinance, where $l$ is the number of negative vertices in $V(\gamma)$.}
    %Thus the number of $k=l+1$ wraps in \thmref{WRB} is nearly tight.}
    \label{fig:wrapAlmostTight}
\end{figure}

The example in \figref{wrapAlmostTight} shows that the number of wraps used in  \thmref{WRB} is nearly tight.
% The number of $k=l+1$ wraps in \thmref{WRB} is nearly tight, 
% where $l$ is the number of negative vertices in $V(\gamma)$. See \figref{} for
% an example of a curve $\gamma$ that requires $k=l$ wraps.
%
We now
%generalize \thmref{WRB} to
show that wrapping resolves obstinance in either direction
of wrapping.
%The proof is straight-forward and given in \appendixref{proofs}.

\begin{theorem}[Wrapping Resolves Obstinance (General)]\label{thm:WRB2}
Let $\gamma \in \C$ with outer basepoint and set $n = \abs{\gamma}$. Then there are constants $k_-, k_+ \leq n+1$ so that
$\obs(Wr_-^{k_-}(\gamma)) = \obs(\Wr_+^{k_+}(\gamma)) = 0$, 
$\Wr_-^{k_-}(\gamma)$ is a negative interior boundary, and $\Wr_+^{k_+}(\gamma)$ is a positive interior boundary.
\end{theorem}
\begin{proof}
If $\gamma(0)$ is a positive basepoint, then $k_+$ exists directly
by \thmref{WRB}. To prove the existence of $k_-$, consider the
intermediary curve $\gamma'$ obtained from $\gamma$ by performing a
$\text{I}_b$-move on the outer edge to create a negatively oriented
loop that is entirely outer to the curve. Additionally, let us set the
basepoint on $\gamma'$ on the new negatively oriented loop $\gamma_-$. Then
$\gamma'$ has a negative outer basepoint. Here is the key observation:
$Wr_+(\overline{\gamma'}) \cong Wr_-(\gamma)$. Thus, by the existence
of $\tilde{k}_+$ for $\overline{\gamma'}$, we have $k_-$ for $\gamma$.

If $\gamma(0)$ is a negative outer basepoint, then since $k_-$ and $k_+$ exist for the reversal $\overline{\gamma}$, we are done,
as $\overline{Wr_{-}^{k_+}(\gamma)} \cong Wr_+^{k_+}(\overline{\gamma})$ and $\overline{Wr_+^{k_-}(\gamma)} \cong Wr_-^{k_-}(\overline{\gamma})$.
\qed
\end{proof}

Let us make a simple observation: once $Wr_{\pm}^k(\gamma)$ is an
interior boundary, so too is $Wr_{\pm}^j(\gamma)$ for any integer
$j \geq k$. This holds because we can simply add the extra $j-k$
wraps to the \SO\ decomposition $\Omega$ of $Wr_+^k(\gamma)$.
%, since
%they are positive/negative Jordan curves and hence are automatically
%positive/negative \SO.
Consequently, \thmref{WRB2} implies that interior boundaries are the
equilibrium point for plane curves with respect to the action of
wrapping. No matter where we begin, we will always eventually land and
stick within the set of interior boundaries.

%\section{Special Self-Overlapping Curves}
\subsection{Irreducible and Strongly Irreducible Curves}\label{sec:special}

We are now ready to apply \thmref{WRB2} to prove sufficient combinatorial conditions for a
curve $\gamma$ to be \SO\ based on
$\whit{\gamma}$ and properties of its direct splits.
If $\gamma\in\C$ has
no proper positive \SO\ direct splits, we call $\gamma$ \myemph{irreducible}.
A special case of irreducibility is of particular interest to us:
If $\whit{\gamma_v} \leq 0$
for all proper direct splits, we call $\gamma$
\myemph{strongly irreducible}.
%
%% If we have a curve of Whitney index 1 that is strongly irreducible, then intuitively the
%% negatively oriented subcurves are near the `top', or `back' of the curve and the
%% positively oriented subcurves are densest near the `bottom', or `front' of the curve.
See \figref{SO_curve}, \figref{decomps}, and $\gamma_{SO}$ in \figref{CurveClasses} for examples of
strongly irreducible curves. Note that a strongly irreducible curve is irreducible
since any positive \SO\ curve $\gamma$ has $\whit{\gamma} = 1$.

We need one simple lemma before proving irreducible curves are \SO.

\begin{lemma}[Existence of a Direct Split]\label{lem:FSD_DS}
Let $\gamma \in \C$ and $\Omega$ be a free subcurve decomposition of $\gamma$, with $|\Omega| \geq 2$. Then $\Omega$ contains a proper direct split. 
\end{lemma}

\begin{proof}
A leaf $v_i$ in the tree $T_{\Omega}$ necessarily corresponds to the basepoint of a direct split $\gamma_i$ in the decomposition $\Omega$. 
Since $|\Omega| \geq 2$, this direct split $\gamma_i$ must be proper. 
\qed
\end{proof}

\begin{theorem}[Irreducible Curves are Self-Overlapping]\label{thm:ISO}
Assume $\gamma$ has $\whit{\gamma} = 1$ and positive outer
basepoint. If $\gamma$ is irreducible, then it is \SO. \end{theorem}

\begin{proof}
Apply \thmref{WRB} to find a $k \in \Z$ such that $Wr_+^k(\gamma)$ is a positive interior boundary.
We know from \propref{EIB:iii} of \thmref{EIB} that there is a \SO\ decomposition
 $\Omega$ of $Wr_+^k(\gamma)$ into positive \SO\ subcurves.
By \lemref{FSD_DS} we know that $\Omega$ must have a \SO\ direct split of
$Wr_+^k(\gamma)$, and we will show that $\gamma$ is the only direct split of $Wr_+^k(\gamma)$
that can be \SO. 

Let $w_i$ be the vertex created by the $i^{th}$ wrap.  The
intersection sequence of $Wr_+^k(\gamma)$ therefore has the prefix
$w_k,w_{k-1},\ldots,w_1$.
Then the direct split $Wr_+^k(\gamma)_{w_i}$ at
$w_i$ on $Wr_+^k(\gamma)$ has $\whit{Wr_+^k(\gamma)_{w_i}}=1+(i-1) = i$ by \lemref{WHSSD}, and
is therefore not \SO\ for $i\geq 2$.
And any direct split $Wr_+^k(\gamma)$ at a vertex of $\gamma$
which is also a proper direct split on $\gamma$ cannot be \SO\ since $\gamma$ is irreducible.
Note that by our notation $w_1$ is the vertex corresponding to the
original basepoint $\gamma(0)$. And this is the only vertex at which the
direct split $Wr_+^k(\gamma)_{w_1}=\gamma$ could potentially be \SO. Thus, it follows
with \lemref{FSD_DS} that $\gamma$ is \SO.
\qed
\end{proof}

\begin{corollary}
Let $\gamma$ have $\whit{\gamma} = 1$ and positive outer
basepoint. Then if $\gamma$ is not \SO, it has a positive \SO\ direct
split $\gamma_v$.
\end{corollary}

We now have a nice corollary: conditions on the Whitney indices of a curve and its subcurves alone can
be sufficient for \SO ness.

\begin{theorem}[Strongly Irreducible Curves are Self-Overlapping]\label{thm:THSO}
Assume $\gamma$ has $\whit{\gamma} = 1$ and positive outer basepoint. If $\gamma$ is strongly irreducible, then it is \SO.
\end{theorem}

Note that strongly irreducible curves are a proper subset of irreducible curves, see $\gamma_{I}$ in \figref{smallerCurveLattice}.
And \thmref{THSO} is false without the basepoint assumption, see \figref{THnotSO}.

\begin{figure}[ht]
\centering
\includegraphics[height = .27\textwidth]{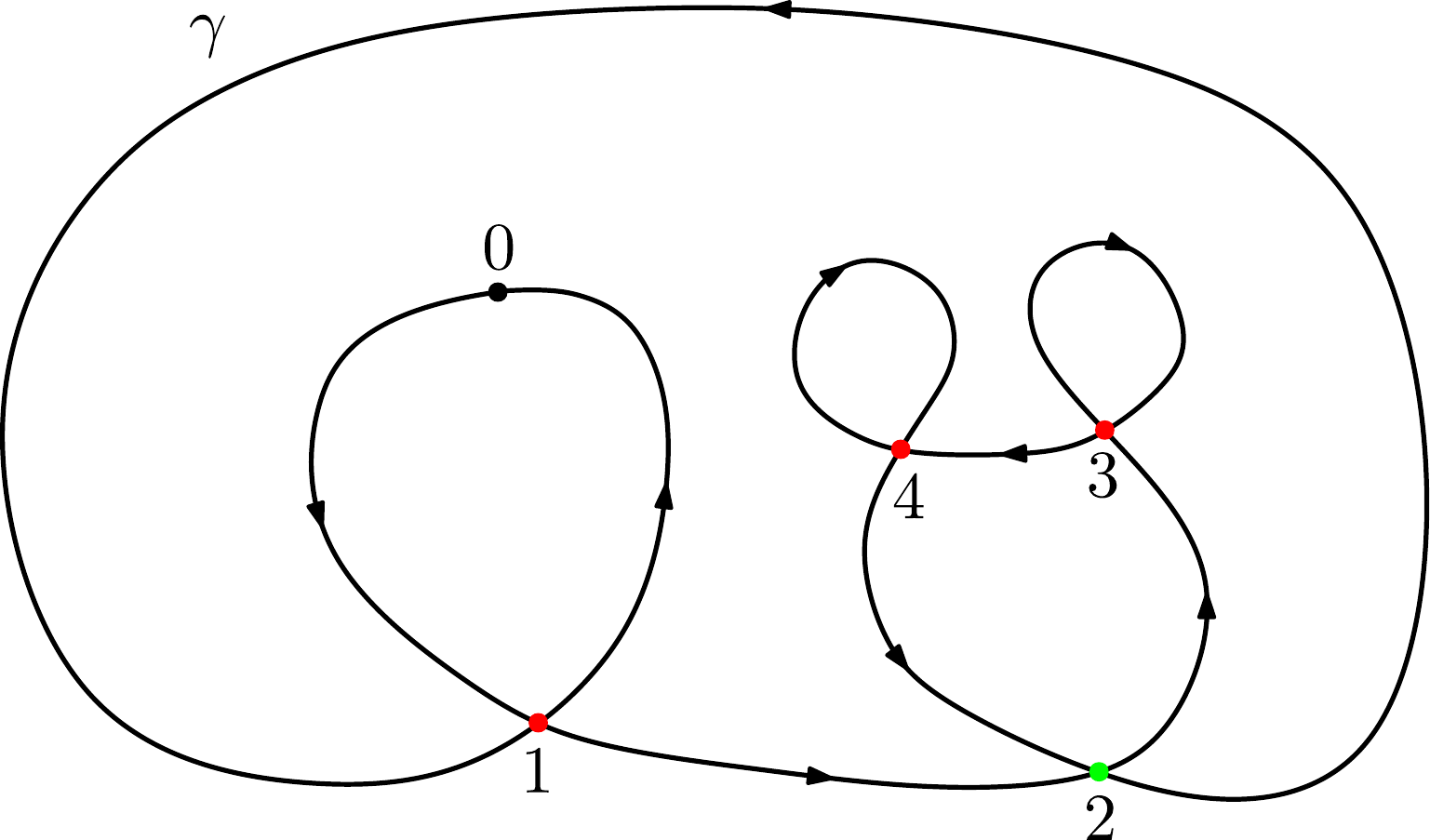}
\caption{This curve $\gamma$ does not have an outer basepoint. It is not \SO,
yet $\gamma$
is strongly irreducible due to the empty positively oriented loop on the indirect split $\gamma_{1^*}$.}
\label{fig:THnotSO}
\end{figure}

One can decide whether a piecewise linear curve $\gamma$ is (strongly) irreducible by
checking the required condition for each direct split.
Let $N$ be the number of line segments of $\gamma$ and $n=|\gamma|=|V(\gamma)|\in O(N^2)$. 
Then irreducibility can be tested in $O(nN^3)$ time, using Shor and Van
Wyk's algorithm to test for \SO ness in $O(N^3)$
time \cite{SVW92}.
Strong irreducibility can be decided in $O(n^2)$ time by
applying \corref{computeWI} to each direct split of $\gamma$.

%% \subsection{Decision Algorithm for Strong Irreducibility}\label{sec-algo}
%% We now describe an algorithm to decide whether a curve is strongly
%% irreducible.
%% %
%% Let $\gamma\in\C$ be a polygonal curve consisting of a sequence of $N$
%% directed line segments $s_1,\ldots,s_N$. Let
%% $G(\gamma)=(V(\gamma),E(\gamma))$ be its induced directed graph, where
%% the vertices $V(\gamma)=\{p_0, p_1, \ldots, p_n\}$ are labeled in the
%% order they appear on $\gamma$, and edge directions are induced by the
%% line segments in $\gamma$. Then $n\in O(N^2)$.
%% %
%% Each vertex $p_i$, $i=1,\ldots,n$, has two incoming edges and two outgoing edges.
%% Let $e_{i,{\rm in}}$ and $e_{i,{\rm out}}$ be the incoming and outgoing edges, respectively, through which $\gamma$ encounters $p_i$ for the first time, at time $t_i$.
%% %
%% And let $e^*_{i,{\rm in}}$ and $e^*_{i,{\rm out}}$ be the incoming and outgoing edges, respectively, through which $\gamma$ encounters $p_i$ for the second time, at time $t^*_i$.
%% %
%% Assume these edges are stored in counter-clockwise order around $p_i$.
%% %

%% For each direct split $\gamma_i$, based at vertex $p_i$, we construct its \decomp\ $\Omega_i$ and compute  $\whit{\gamma_i}$ as the number of positively oriented loops minus the number of negatively oriented loops in $\Omega_i$ according to \lemref{SSDsign}.
%% %
%% We store $\Omega_i$ in a tree

\subsection{Global Balanced Loop Insertion}

We now introduce an operation called balanced loop insertion, which is complementary to the
balanced loop deletion applied in the proof of \thmref{WRB}. We show in \thmref{GBLI} that any curve $\gamma$ with positive outer basepoint
and $\whit{\gamma} = 1$ can be transformed into a \SO\ curve, more specifically, a strongly irreducible \SO\ curve,
through a sequence of balanced loop insertions. This result is a nice parallel to \thmref{WRB}.

%\begin{definition}[Balanced Loop Insertion]\label{def:BLI}
Let $\gamma \in \C$ have a positive outer basepoint. Then given
any edge $e$ from $G(\gamma)$, we define \myemph{balanced loop insertion} on
$\gamma$ with respect to $e$ as follows: First, perform a
$\text{I}_b$-move to insert a negatively oriented loop, on the right
side of the edge $e$, and smooth the resulting curve $\gamma'$ until
it is normal and regular. Then, wrap around $\gamma'$ to create
$\gamma''= Wr_+(\gamma')$.
%We call $\gamma''$ the curve obtained through
%balanced loop insertion on $\gamma$ with respect to $e$.
See \figref{BLI}.
%\end{definition}
%
\begin{figure}[ht]
\centering
\includegraphics[height = .36\textwidth]{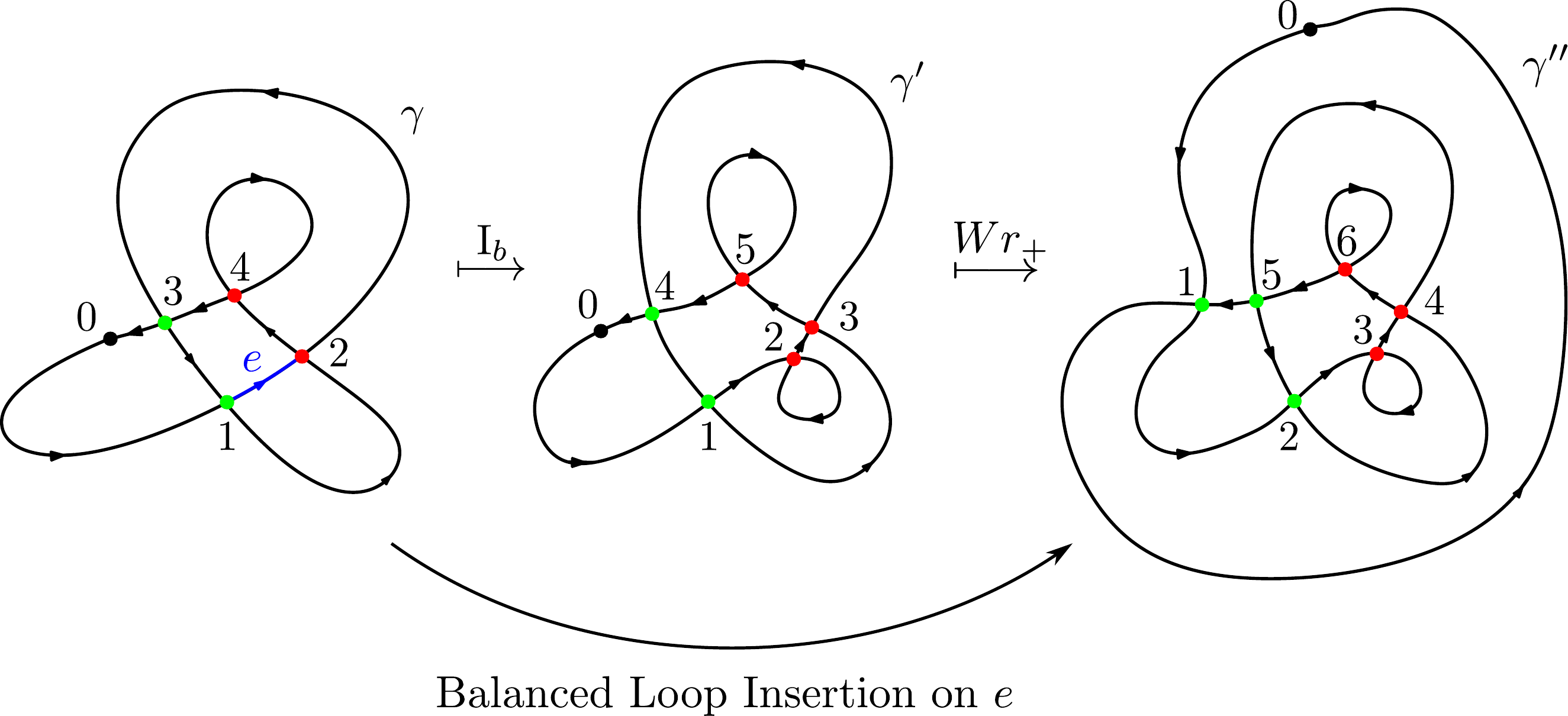}
\caption{Balanced loop insertion on a \SO\ curve $\gamma$ with respect to an edge $e$. Note that the curve produced, $\gamma''$, is also \SO.}
\label{fig:BLI}
\end{figure}
%
%We are particularly interested in applying balanced loop
%insertion to curves with $\whit{\gamma} = 1$. In fact,
%we now define an iterative process of balanced loop insertion which will transform any such curve, with positive outer basepoint,
%into a strongly irreducible, \SO\ curve.
The following is an interesting property of balanced loop insertion.

\begin{lemma}
Let $\gamma$ have $\whit{\gamma} = 1$ and positive outer basepoint. If $\gamma$ is strongly irreducible,
and $\gamma'$ is obtained from $\gamma$ by balanced loop insertion, then $\gamma$ is strongly irreducible as well.
\end{lemma}

\begin{proof}
Let $i: V(\gamma) \rightarrow V(\gamma')$ be the inclusion map, sending vertices on $\gamma$ to the corresponding
vertices on $\gamma'$. Then the only possible change to the direct split was that we added a negatively oriented
loop, so
$\whit{\gamma'_{i(v)} } \leq \whit{\gamma_v} \leq 0$.
The only other direct splits we need to check are those of the vertices $u, v$ created by balanced loop insertion.
We note immediately that $\whit{\gamma'_u} = -1$ where $u$ is the basepoint of the new negatively oriented loop.
Meanwhile, if $v$ is the vertex created by the wrap, then $\whit{\gamma'_{v}} = \whit{\gamma} -1 = 0 $.
Hence, $\gamma'$ is strongly irreducible.
\qed
\end{proof}

%\subsection{Global Balanced Loop Insertion}

While strong irreducibility is preserved by balanced loop insertion, and
consequently, so too is \SO ness, if $\gamma$ has a positive outer basepoint
and $\whit{\gamma} = 1$, we need a stronger operation to transform a
non-\SO\ curve into a \SO\ curve.
%We define this operation now.

%\begin{definition}[Global Balanced Loop Insertion]\label{def:GBLI}
Let $\gamma\in \C$. \myemph{Global balanced loop insertion}, denoted by $M: \C \rightarrow \C$, 
applies balanced loop insertion simultaneously once on every edge of $\gamma$.
%\end{definition}
%
Since there are $2\abs{\gamma} + 1$ edges on $G(\gamma)$, the operator $M(\cdot)$ applies balanced loop insertion $2\abs{\gamma}+1$ times.  
Equivalently, $M(\gamma)$ can be obtained by performing a $\text{I}_b$-move to the right of every edge of $\gamma$, adding a
new negatively oriented loop, and then wrapping the curve $2\abs{\gamma}+1$
times.
See \figref{GBLI}  for an example of global balanced loop insertion.
%, which transforms a non-\SO\ curve into a \SO\ curve.
%

\begin{figure}[ht]
\centering
\includegraphics[height=.33\textwidth]{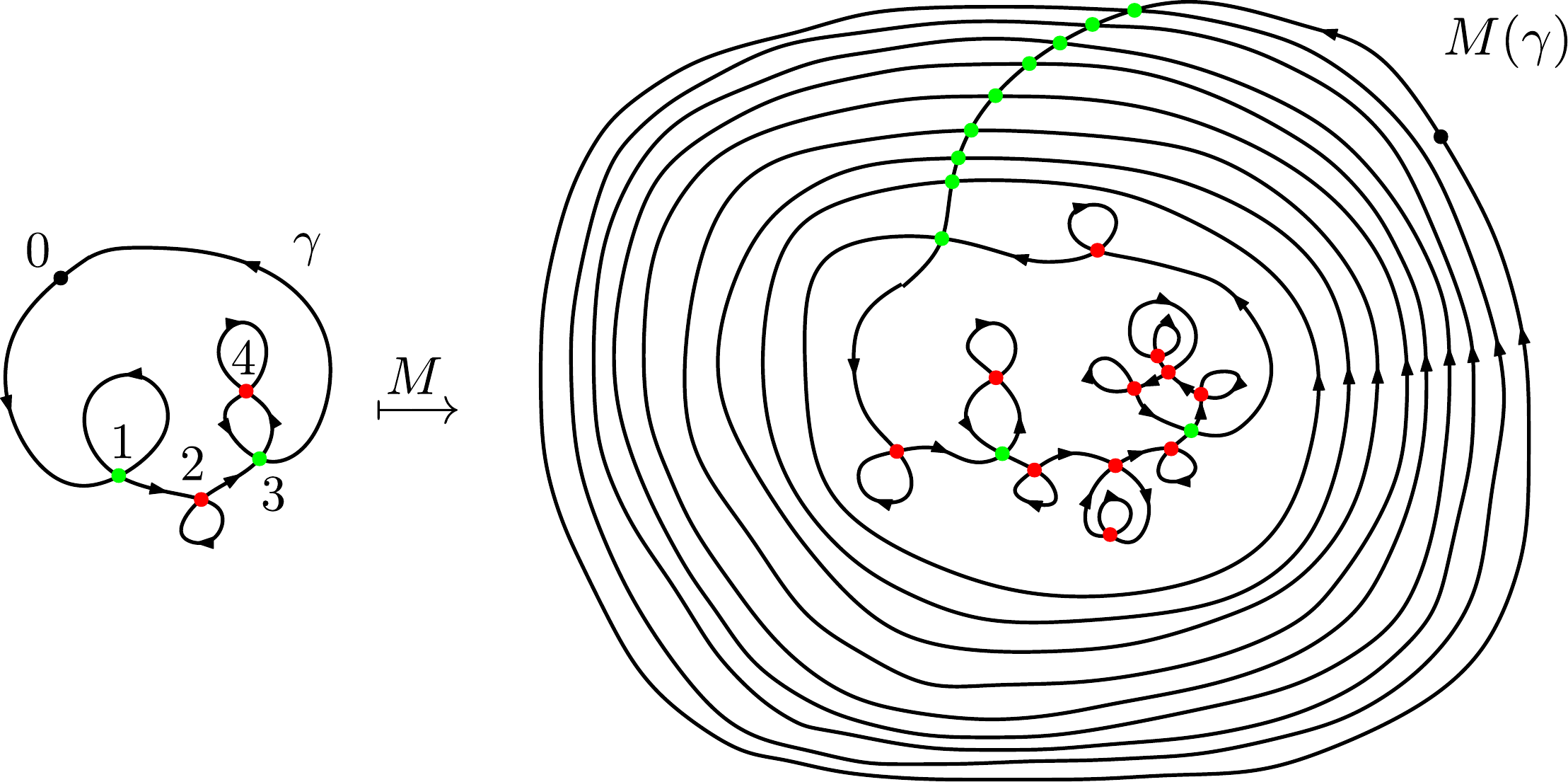}
\caption{Global balanced loop insertion applied to a  curve $\gamma$. Since $\gamma$ has empty positively oriented loops it is not \SO (by \lemref{EmptyLoop}).
The curve $M(\gamma)$ is strongly irreducible and \SO\ by \thmref{GBLI}. }
\label{fig:GBLI}
\end{figure}

We now show that this operation transforms any curve $\gamma$ with
positive outer basepoint and $\whit{\gamma} = 1$ into a strongly irreducible curve
and hence a \SO\ curve by \thmref{THSO}.

\begin{theorem}\label{thm:GBLI}
Let $\gamma$ have positive outer basepoint and $\whit{\gamma} = 1$. Then $M(\gamma)$ is strongly irreducible and \SO.
\end{theorem}

\begin{proof}
Let $\gamma$ be a curve with positive outer basepoint. We will utilize
 the identity $\whit{\gamma} = \sum_{v \in V(\gamma)} \sgn(v)$ shown by
 Titus and Whitney \cite{TitusNC,WH37}; note the inclusion of the
 basepoint here.
It
 follows immediately that $\whit{\gamma} \leq \abs{\gamma} + 1$ for
 $\gamma \in \C$ with outer basepoint. While a direct split $\gamma_v$
 may not have an outer basepoint, we can instead consider a curve
 $\gamma_v'$ with the same image and orientation as $\gamma_v$, and
 hence the same number of intersection points, and an outer
 basepoint. We then have $\whit{\gamma_v}
 = \whit{\gamma_v'} \leq \abs{\gamma_v'}+1= \abs{\gamma_v} +1$. On the
 other hand, since every edge of $\gamma_v$ received at least one
 negatively oriented loop, as edges of $\gamma_v$ may be further
 subdivided on $\gamma$, we note that we inserted at least
 $2\abs{\gamma_v} + 1$ negatively oriented loops on the direct split
 $\gamma_v$.
Thus, if we write $i : V(\gamma) \rightarrow V(M(\gamma))$
for the natural inclusion map and set $\tilde{v} = i(v)$, then we see $ \whit{M(\gamma)_{\tilde{v}}} \leq -\abs{\gamma_v} \leq 0$.
 Hence, all the vertices on $M(\gamma)$ that
 came from $\gamma$ will not yield direct splits of Whitney index 1 or
 greater.

Now, the only other vertices to consider are the basepoints of the new
 negatively oriented loops and the basepoints of the wraps. Clearly,
 for any vertex $u$ of the former kind, we have
 $\whit{M(\gamma)_u}=-1$ for the direct split at $u$.
We now address the basepoints of the wraps. Let
$M(\gamma)_i$ be the direct split and $M(\gamma)_{i^*}$ be the indirect split 
at the $i$-th vertex of $M(\gamma)$. 
%
%Let $w_i$ denote the $i$-th vertex on $M(\gamma)$ and let 
By definition $M(\gamma)$ contains $2|\gamma|+1$ outer wraps, which implies
$\whit{M(\gamma)_{i^*}} = i$  for all $i \in \{1, \dots, 2\abs{\gamma}+1 \}$.
 And since direct splits
 and indirect splits are complementary, it follows from \lemref{WHSSD} that
$\whit{M(\gamma)}=\whit{M(\gamma)_i}+\whit{M(\gamma)_{i^*}}$ and hence
%% and $\whit{C} = \whit{C_v}
%%  + \whit{C_{v^*}}$ for any curve $C$ and vertex $v \in V(C)$, we have
 $\whit{M(\gamma)_i} = 1-i \leq 0$ for any $i \in \{1, \ldots, 2\abs{\gamma}+1\}$. Thus, $M(\gamma)$ is indeed strongly irreducible.
\qed
\end{proof}

%In a sense, the curves $M(\gamma)$ produced by global balanced loop insertion are paragons of strongly irreducible curves.
%All of the negatively oriented subcurves lie at the `back' or `top' of the curve, stuffed inside the direct split corresponding to the original curve.

\section{Discussion}
We introduced new curve classes (zero-obstinance, irreducable, and
strongly irreducable curves; see \figref{smallerCurveLattice}), which
help us understand \SO\ curves and interior boundaries. We proved
combinatorial results and showed that wrapping a curve resolves
obstinance. These new mathematical foundations for \SO\ curves and
interior boundaries could pave the way for related algorithmic
questions.  For example, is it possible to decide whether a curve is
\SO\ in $o(N^3)$ time? How fast can one decide \SO ness of a
curve on the
sphere?
Can one decide irreducability in $o(n^2)$ time, even in the presence of a large number of linked subcurves?

\section{Acknowledgments}
We thank Brittany Terese Fasy for fruitful discussions, in particular
in the early stages of the research, as well as for proof-reading.
This paper is based on the Honors thesis of the first
author \cite{Parker}. Much of this research was enabled by the use of
a computer program \cite{Program} which can determine whether a plane curve is \SO,
compute it's minimum homotopy area, and display the \SO\ decomposition associated
with a minimum homotopy. \figref{Ideal} and \figref{SPD} were created with this program.

\bibliographystyle{plainurl}
\bibliography{bibliography}

\appendix

\section{Whitney Indices \& Loop Decompositions}\label{appendix:whit}

In this section are interested in a refinement of \SO\ decompositions. Let $\Omega =(\gamma_i)_{i=1}^n$ 
be a \SO\ decomposition. If each $\gamma_i$ is simple,  i.e. a loop,
then we call $\Omega$ a \emph{loop decomposition}. See \figref{SSD} for an example.

\begin{figure}[h]
        \centering
        \includegraphics[height=.35\textwidth]{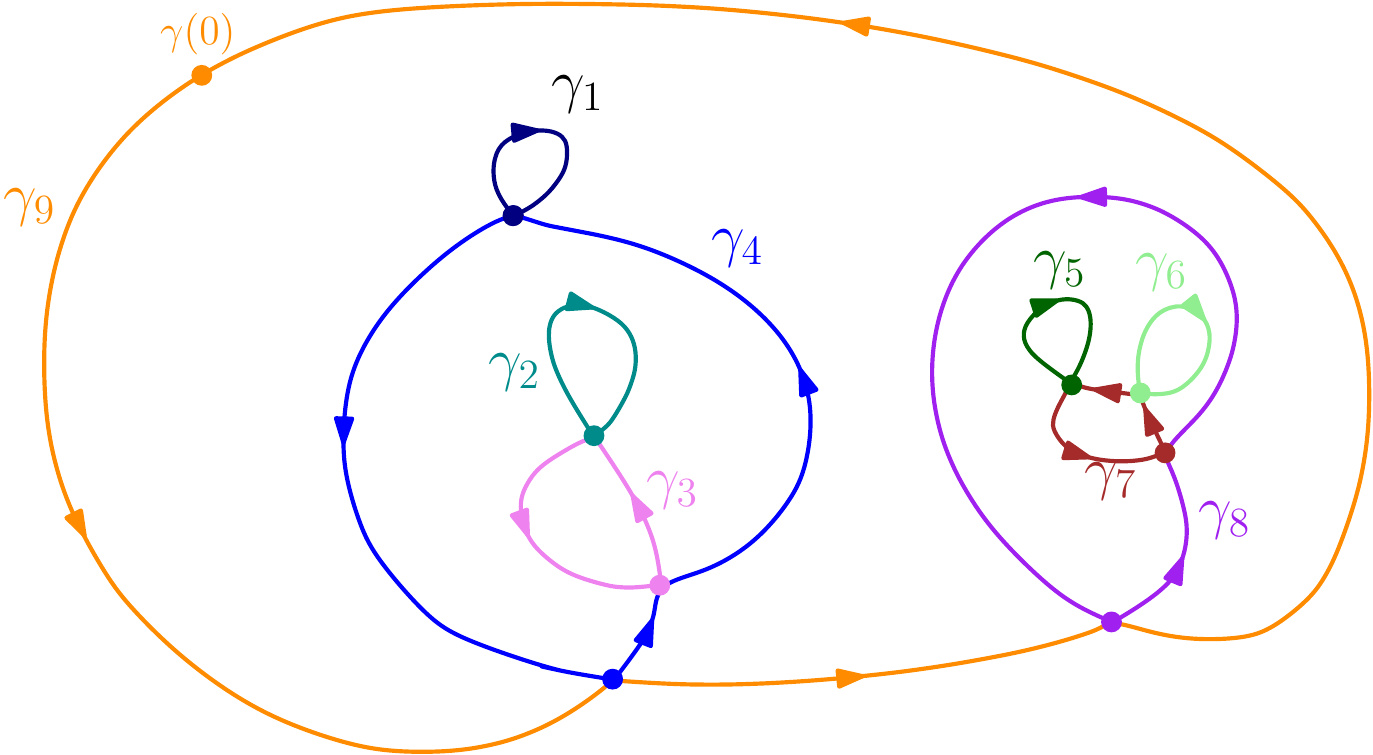}
        \caption{A loop decomposition of a curve $\gamma$.}
        \label{fig:SSD} 
\end{figure}

We now recall a construction of Seifert \cite{Gauss,Seifert}. He
introduced a decomposition of plane curves by performing so-called
``uncrossing moves", which essentially split a curve $\gamma$ at a
vertex $v$ into the two (nearly) disjoint pieces $\gamma_v$ and
$\gamma_{v^*}$. See \figref{UncrossingMove}. If we cut the two curves
at the vertex $v$, and smooth them,
we obtain two completely disjoint plane curves. Iterating this process
across many vertices of $\gamma$, one can achieve a decomposition of
the original curve $\gamma$ into a set of Jordan curves.

\begin{figure}[ht]
	\centering
	\includegraphics[width=.42\textwidth]{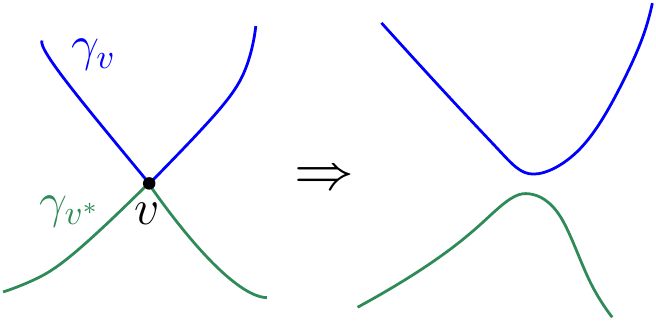}
	\caption{An uncrossing move applied to a vertex $v$, splitting the curve into two pieces.}
	\label{fig:UncrossingMove}
\end{figure}

Suppose $\gamma$ is decomposed into Jordan curves $\{C_i\}_{i=1}^k$
using these uncrossing moves. Then Seifert and Gauss \cite{Gauss,Seifert} showed that
$\whit{\gamma} = \sum_{i=1}^k \whit{C_i}$. Using our terminology and replacing
the Jordan curves with loops in a loop decomposition, we obtain the
equivalent fact:

\begin{lemma}[Compute Whitney Index]\label{lem:SSDsign}
Let $\gamma \in \C$, let $\Omega$ be a \decomp\ of $\gamma$, and let $n_+$ be the
    number of positively oriented loops and~$n_-$ be the number of negatively oriented loops in $\Omega$. Then $\whit{\gamma} = n_+-n_-$.
\end{lemma}

As a consequence of this lemma, we can prove linearity of Whitney indices across a direct split decomposition. 

\WHSSD*
%% \noindent
%% {$\blacktriangleright$\nobreakspace\sffamily\bfseries \lemref{WHSSD}} (Whitney Index Through Decompositions).
%% {\em  Let $\gamma \in \C$ and $\Omega$ be a direct split decomposition of $\gamma$.
%%  Then $\whit{\gamma} = \sum_{C \in \Omega} \whit{C}$.
%% }

 \begin{proof}
The key is to piece together \decomp's of each subcurve $C \in \Omega$. Write $\Omega = (C_i)_{i=1}^k$ and take $\Psi_i$ to be a loop decomposition of $C_i$. Then $\Psi = (\Psi_1, \dots, \Psi_k) $ is a loop
 decomposition of $\gamma$. Since $\whit{C_i} = p_i - n_i$ where $p_i$ is the number of positively oriented
 subcurves in $\Psi_i$ and $n_i$ the number of negatively oriented subcurves, we conclude that
 \[\whit{\gamma} = \sum_{i=1}^k p_i - \sum_{i=1}^k n_i = \sum_{i=1}^k (p_i - n_i) = \sum_{i=1}^k \whit{C_i}\; .\]
 \qed
\end{proof}

%%%%%%%%%%%%%%%%%%%%%%%%%%%%%%%%%%%%%%%%%%%%%%%%%%%%%%%%%%%%%%%%%%%%%%%%%%%%%%%%%%%%%%%%%%%%%%%

\section{Lattice}
\label{appendix:lattice}
We introduce two more classes of curves.
We say a face $F$ is \myemph{good} when its depth is equal to its winding
number. If a curve $\gamma\in\C$ is positive consistent and all faces
on $G(\gamma)$ are good then we call the curve \myemph{good}.
%
%
%% It is quite easy to see inductively that good curves are indeed
%% interior boundaries, by \thmref{EIB} (iii). Indeed, contracting a loop
%% on a good curve leaves a remaining good curve. In fact, this result is
%% equivalent to Theorem 5 of \cite{Marx}.  Now, let us introduce one
%% final definition before we show the relations between all the classes
%% of curves defined through minimum homotopy area.
%
We call a curve \myemph{basic} if all of its \SO\ decompositions are \decomp s. That is, the only
\SO\ decompositions are decompositions into loops.
By \thmref{Ideal}, basic curves with zero obstinance can be decomposed into good curves.

We define {\sc Simple, Basic-Zero-Obstinance, Zero-Obstinance} as
classes of those curves  that have the property described by the class
name. The classes {\sc SO$^+$, Interior-Boundary$^+$, Consistent$^+$, Good$^+$} consist of the curves  with
the {\em positive} property described by the class name (positive \SO, positive interior boundary, positive consistent, and good curves that are positive consistent).
\figref{CurveClasses} and \thmref{CurveClasses} show the relationship between
these curve~classes.
The curves in \figref{CurveClasses} show that the inclusions
in parts ~\ref{CC:1},~\ref{CC:2}, \ref{CC:3} of \thmref{CurveClasses} are proper.
We first state a lemma that we need in the proof of the theorem.

\begin{lemma}[Negatively Oriented Loop]\label{lem:negloop}
  Let $\gamma \in \C$ be a non-simple, positive \SO\ curve with positive outer basepoint. Then $\gamma$ has a negatively oriented loop $\gamma_v$.
\end{lemma}

\begin{proof}
Let $\mathcal{V} = \{ \, v \in V(\gamma) \, | \, \gamma_v \text{ is a loop}\, \}$. 
By \lemref{OL}, we know $\mathcal{V} \neq \emptyset$. 
Define the relation $v \prec w$ for $v,w \in
 \mathcal{V}$ whenever $\intt( v) \subseteq \intt( w)$. It is
  straightforward to verify that $(\mathcal{V}, \prec)$ is a poset. Of
 course, since $|\mathcal{V}| $ is finite, we can choose $v_0 \in
 \mathcal{V}$ minimal with respect to $\prec$.  Now, suppose that
 $\gamma_{v_0}$ were a positive loop. We show this is a contradiction
 to complete the proof. Indeed, by minimality of $v_0$, there are no
 loops of $\gamma$ completely contained inside
 $\intt{\,\gamma_{v_0}}$. This means every time a strand of $\gamma$
 crosses from outside to inside $\gamma_w$, the strand does not cross
 itself inside of $\gamma$. Topologically, then, $\overline{\intt( \gamma_{v_0} )}$
 looks like a disk with finitely many simple arcs
 traveling from boundary to boundary.  By way of a (regular) right
 sense-preserving homotopy, we can sweep each such arc until it no
 longer intersects $\gamma_{v_0}$. As we need only sweep finitely many
 such arcs, let us denote $\gamma'$ as the result of this process. By
 \lemref{RSP_homotopy}, we know $\gamma'$ is \SO. On the other hand,
 $\gamma'$ has an empty positive loop, namely the one we just emptied,
 which contradicts \lemref{EmptyLoop}.

\qed
\end{proof}

\begin{theorem}[Curve Classes]
%[Curve Classes Through Minimum Homotopy Area]
\label{thm:CurveClasses}
If $\gamma \in \C$ has a positive outer basepoint, then:
\begin{enumerate}
\item {\sc Simple} $\subset$ {\sc SO$^+$} $\subset$ {\sc Interior-Boundary$^+$} $\subset$ {\sc Consistent$^+$}\label{CC:1}
\item {\sc Simple} $\subset$ {\sc Good$^+$} $\subset$ {\sc Basic-Zero-Obstinance} $\subset$ {\sc Zero-Obstinance}\label{CC:2}
\item {\sc Good$^+$} $\subset$ {\sc Interior-Boundary$^+$} $\subset$ {\sc Zero-Obstinance}\label{CC:3}
\item {\sc Consistent$^+$} $\cap$ {\sc Zero-Obstinance} $=$ {\sc Interior-Boundary$^+$}\label{CC:4}
%$\gamma$ is an interior boundary iff it is consistent and zero badness.
\item {\sc Consistent$^+$} $\cap$ {\sc Basic-Zero-Obstinance} $=$ {\sc Good$^+$}\label{CC:5}
%$\gamma$ is good iff it is basic zero obstinance and consistent iff it is basic zero obstinance and an interior boundary.
\item {\sc SO$^+$} $\cap$ {\sc Good$^+$} $=$ {\sc Simple}\label{CC:6}
%$\gamma$ is simple iff it is SO and good.
\end{enumerate}
\end{theorem}

\begin{proof}
By definition, a simple curve with a positive outer basepoint is positive \SO, and $+k$-boundaries are positive consistent. By \thmref{ESO}, a positive \SO\ curve is a $1$-boundary, which proves \ref{CC:1}.
A simple curve is trivially good.
By \lemref{minHomotopy_and_windingAreas} and \lemref{minHomotopy_and_depth},
good curves have zero obstinance.
We now show that good curves are basic by contradiction.
    Suppose~$\gamma$ admitted a \SO\ decomposition $\Omega= (\gamma_i)_{i=1}^k$
with a non-simple \SO\ curve $\gamma_j$. Then~$\gamma_j$ must contain a
    negatively oriented loop by \lemref{negloop}.
But we could then create a finer decomposition of~$\gamma$
by decomposing~$\gamma_j$ into loops. Precisely, let
$\Psi$ be a loop decomposition of $\gamma_j$ and consider $\Gamma = (\gamma_1, \dots, \gamma_{j-1}, \Psi, \gamma_{j+1}, \dots, \gamma_k)$.
Then $\Gamma$ is a free subcurve decomposition of $\gamma$ refining $\Omega$. Let $C$ be
a negatively oriented loop in~$\Psi$ and take any face $F$
contained in the interior of $C$. Now, take a path $P$ from~$F$ to the
exterior face on~$G(\gamma)$ such that the depth is monotonically
decreasing along the path~$P$. Since $F$ is contained inside $C$, the path
$P$ must cross $C$ to reach $F_{ext}.$ However, when $P$ crosses past $C$,
we see the depth either decrease by 1 or remain unchanged, while the winding number increases by 1, 
since $C$ is negatively oriented. 
We learn that $wn(F,\gamma)<D(F,\gamma)$, which is a contradiction.
This proves~\ref{CC:2}, with the last inclusion being trivial.
Since a good curve is basic and positive consistent, it is a positive interior boundary by \propref{EIB:iii} of \thmref{EIB}. And interior boundaries have zero obstinance by definition, which proves \ref{CC:3}.

By definition, interior boundaries are consistent and have zero obstinance,
    which proves~\ref{CC:4}.
%
% {\sc Consistent} $\cap$ {\sc Basic-Zero-Obstinance} $=$ {\sc Good}\label{CC:5}
If a curve $\gamma$ is basic and positive consistent, it follows that $\gamma$ admits
a \decomp\ $\Omega$ with only positively oriented subcurves. Since by \obsref{wnDecomp}
$wn(F,\gamma)=\sum_{\gamma_i \in \Omega} wn(F,\gamma_i)$ and each $wn(F,\gamma_i)\in\{0,1\}$,
we must have $wn(F,\gamma_i) \geq D(F,\gamma_i)$. The bound
$wn(F, \gamma) \leq D(F, \gamma)$ holds generally. Thus $\gamma$ is good, which together with \ref{CC:1} and \ref{CC:2} proves \ref{CC:5}.
By \lemref{negloop}, a good curve that is \SO\ may not have a negatively oriented loop
and hence must be
simple, which together with \ref{CC:1} and \ref{CC:2} proves~\ref{CC:6}.
\qed
\end{proof}

\begin{figure}[bhpt]
\center
\includegraphics[width=1\textwidth]{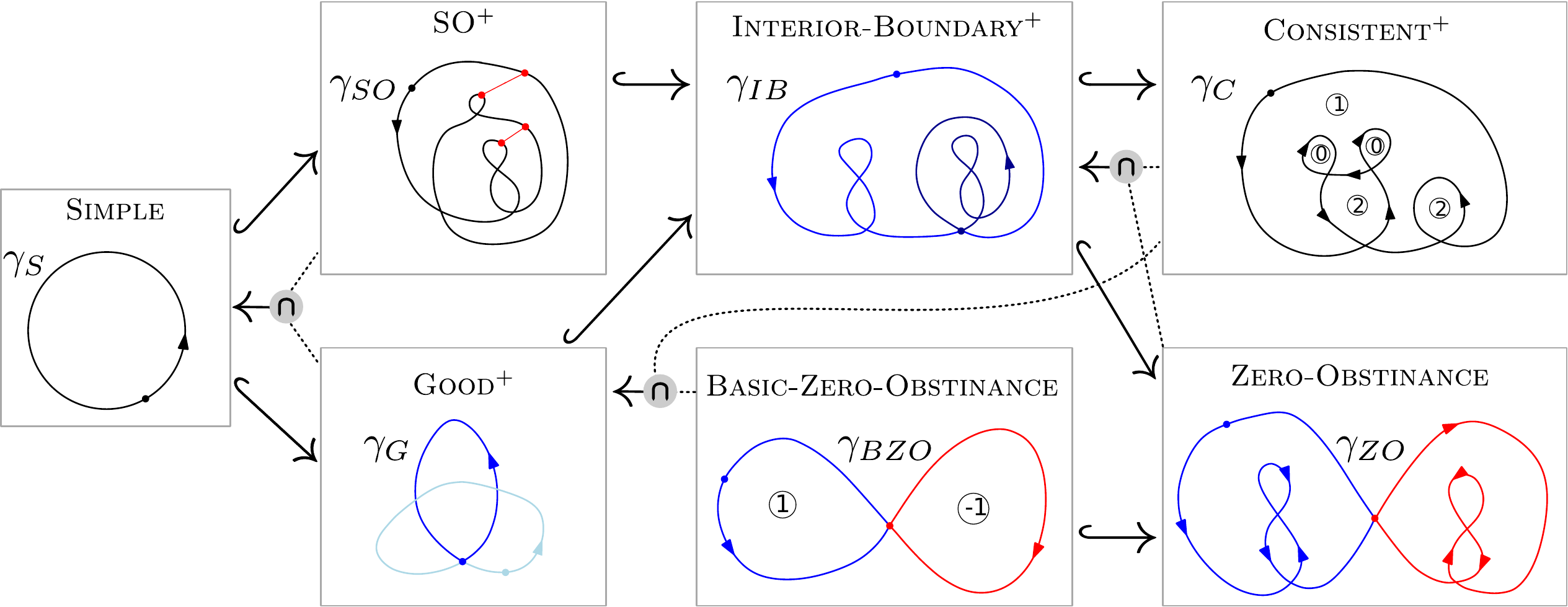}
\caption{Lattice of curve classes.  Solid arrows are inclusions. Two dashed
    lines meet to form an inclusion.
    A member curve of each class is displayed, along
with appropriate information to justify membership in
its curve class: A set of blank cuts of~$\gamma_{SO}$
are shown in red, a \SO\ decomposition of~$\gamma_{IB}$ in blue, the
winding numbers of $\gamma_{C}$, the \decomp s of $\gamma_{G}$ and $\gamma_{BZO}$,
and the \SO\ decomposition of~$\gamma_{ZO}$.
The inclusions are proper:
$\gamma_{BZO}$ is not consistent and consequently not good;
$\gamma_{IB}$ is not good, and since $\whit{\gamma_{IB}} = +2$ it is not \SO;
$\gamma_C$ is not an interior boundary because $\whit{\gamma_C}=1$ but since it has
an empty positively oriented loop it is not \SO;
$\gamma_{ZO}$ is not an interior boundary because it is not consistent.
See \thmref{CurveClasses}.}
\label{fig:CurveClasses}
\end{figure}

\end{document}